\documentclass[a4paper,twocolumn]{quantumarticle}
\pdfoutput=1

\usepackage{amsmath,amssymb,amsthm,mathtools}
\usepackage{graphicx}
\usepackage{booktabs}
\usepackage{enumerate}
\usepackage{hyperref}
\usepackage{xurl}        
\usepackage{xcolor}
\usepackage{cleveref}
\usepackage{microtype}    
\usepackage{array}        
\graphicspath{{./figures/}{./}}

\newtheorem{theorem}{Theorem}[section]
\usepackage{aliascnt}
\newcommand{\sharedthm}[2]{%
  \newaliascnt{#1}{theorem}\newtheorem{#1}[#1]{#2}\aliascntresetthe{#1}}
\sharedthm{proposition}{Proposition}
\sharedthm{lemma}{Lemma}
\sharedthm{corollary}{Corollary}
\sharedthm{definition}{Definition}
\sharedthm{conjecture}{Conjecture}
\sharedthm{observation}{Observation}
\sharedthm{remark}{Remark}

\crefname{theorem}{Theorem}{Theorems}
\crefname{proposition}{Proposition}{Propositions}
\crefname{lemma}{Lemma}{Lemmas}
\crefname{corollary}{Corollary}{Corollaries}
\crefname{definition}{Definition}{Definitions}
\crefname{conjecture}{Conjecture}{Conjectures}
\crefname{observation}{Observation}{Observations}
\crefname{remark}{Remark}{Remarks}

\newcommand{\ket}[1]{|#1\rangle}

\newcommand{\ketbra}[2]{|#1\rangle\!\langle#2|}
\newcommand{\Wfree}{\mathcal{W}_{\mathrm{free}}}
\newcommand{\Stab}{\mathrm{Stab}}
\newcommand{\XL}{X_L}
\newcommand{\YL}{Y_L}
\newcommand{\ZL}{Z_L}
\newcommand{\tr}{\mathrm{Tr}}
\newcommand{\sct}{|\sin\theta|+|\cos\theta|}

\title{A Phase-Space Geometric Measure of Magic in Qubit Systems}

\author[1]{Soumyojyoti Dutta\thanks{m24iqt014@iitj.ac.in}}
\author[1]{Tushar\thanks{p24ph0012@iitj.ac.in}}
\affil[1]{Indian Institute of Technology Jodhpur, Rajasthan 342030, India}

\date{24 March 2026}

\begin{document}
\maketitle

\begin{abstract}
Magic---the resource enabling quantum computational advantage beyond
stabilizer circuits---has a clean phase-space characterization in odd prime
dimensions that qubits notoriously lack.
We study $C(\rho)$, the $\ell_1$ distance from a state's discrete Wigner
function to the stabilizer polytope, and determine its exact geometry.
For a single qubit with Bloch vector $\vec r$ we prove the closed form
$C=\tfrac12\max(0,\|\vec r\|_1-1)$, so the maximum is $(\sqrt3-1)/2$,
attained precisely at the eight face states. For products of single-qubit
states, $C(\rho_1\otimes\cdots\otimes\rho_n)\leq\prod_i(1+C(\rho_i))-1$,
with equality whenever every factor has Bloch sign product
$r_xr_yr_z\leq0$; this includes the exact tensor powers
$((1+\sqrt3)/2)^n-1$ and the exact maximum of $C$ over fully separable
$n$-qubit states.

Because no discrete Wigner function for qubits is Clifford covariant, any
such measure is frame dependent, and we identify which of its features
are not.
For a single qubit there are exactly two valid frames up to phase-space
translation (eight choices of origin) and $C$ is identical on both; the
separable maxima hold in every product frame, since these are related by
local Clifford gates.
Pauli-covariant two-qubit frames number $2048$ up to translation, and
across them the ratio
$C(\rho^{Rx})/C(\rho^{Ry})$ takes the values $\tfrac12$, $1$ and $2$; the
inversion already occurs between product frames of opposite tetrahedral
chirality.
In every one of these frames some stabilizer state has negative Wigner
entries, a consequence of the Mermin--Peres square, so the
two-qubit bound $\Gamma \geq 1 + C/2$ against the robustness of magic
$\Gamma$ is frame independent.
Within the Wootters frame we prove integer values $1,2,1$ of the tightness
ratio $\kappa := (\Gamma-1)/C$ for three families in the two-qubit
repetition codespace, with explicit primal and dual certificates whose
optimal witnesses are logical Pauli operators compressed to the codespace.
For $n$-qubit GHZ phase states we prove $\Gamma=|\sin\phi|+|\cos\phi|$ for
all $n$ and $C\geq 2^{\lfloor n/2\rfloor-1}(|\sin\phi|+|\cos\phi|-1)$,
tight for $n=2$ and, by exact linear programming, for $n=3,4$, so that
$\kappa\leq 2^{1-\lfloor n/2\rfloor}$ decays exponentially: the Wigner distance can overstate simulation cost by an
exponentially large factor.
Finally, $C$ is not a magic monotone, so it cannot replace $\Gamma$ in
distillation bounds.
\end{abstract}

\section{Introduction}
\label{sec:intro}

Understanding the source of quantum computational advantage remains a central
problem in quantum information theory.
Stabilizer circuits, despite exhibiting entanglement and other nonclassical
features, are efficiently classically simulable by the Gottesman--Knill
theorem~\cite{Gottesman1998}.
Universal quantum computation can be achieved by supplementing stabilizer
operations with non-stabilizer resource states, commonly called \emph{magic
states}~\cite{Bravyi2005}.
This raises the fundamental question: what structural property of a quantum
state enables this enhancement?

For systems of odd prime dimension, a satisfying answer exists.
Gross~\cite{Gross2006} showed that the discrete Wigner function built from
the $\mathrm{GF}(d^2)$ phase-space has the Hudson property: pure states have
non-negative Wigner functions if and only if they are stabilizer states.
This connects, in a single equivalence, contextuality, Wigner negativity, and
usefulness for magic-state distillation~\cite{Howard2014}.
The boundary between classical simulability and quantum advantage admits a
clean geometric and operational characterization.

\paragraph{The qubit problem.}
The qubit case ($d=2$) is structurally more subtle.
State-independent contextuality exists~\cite{Mermin1990}, but discrete Wigner
representations for qubits do not yield a direct equivalence between
non-negativity, non-contextuality, and the stabilizer
polytope~\cite{Raussendorf2017,Delfosse2015,Raussendorf2020}.
For the Gibbons--Hoffman--Wootters family of qubit Wigner
functions~\cite{Gibbons2004}, Galv\~{a}o showed that the states with
non-negative Wigner function in \emph{every} member of the family are
exactly the stabilizer mixtures for Hilbert-space dimension below
six~\cite{Galvao2005}, and Cormick \emph{et al.}\ proved that the only pure
states non-negative in every member are stabilizer states~\cite{Cormick2006};
in any \emph{single} frame, however, positivity and stabilizerness
diverge (see also Ref.~\cite{KociaLove2017}).
For two or more qubits, stabilizer states themselves have negative Wigner
entries in a fixed frame~\cite{Raussendorf2017,Raussendorf2020}, so the
free region in a fixed frame is not the
non-negative orthant but the convex hull of stabilizer Wigner
functions---a polytope whose geometry is richer and less transparent.

Several measures of magic have been proposed for qubits.
The \emph{robustness of magic} $\Gamma(\rho)$~\cite{Howard2017} and its
regularization are operationally grounded: the classical cost of
quasiprobability simulation of a circuit using state $\rho$ scales as
$\Gamma(\rho)^2$.
The \emph{mana}~\cite{Veitch2014} $\mathcal{M}(\rho) = \log\|W_\rho\|_1$
provides a simpler scalar summary.
The \emph{stabilizer R\'{e}nyi entropies}~\cite{Leone2022} offer a
computable family, monotone for R\'enyi index
$\alpha\geq2$~\cite{LeoneBittel2024}.
Further monotones---the stabilizer extent~\cite{Bravyi2019}, generalized
robustness and dyadic negativity~\cite{Seddon2021}, the thauma~\cite{WangWildeSu2020},
and the many-body measures of Ref.~\cite{LiuWinter2022}---each trade
computability against operational meaning differently.
The link between phase-space $\ell_1$ norms and simulation cost goes back to
quasiprobability sampling: estimating outcome probabilities costs a number
of samples quadratic in the $\ell_1$ norm of the
quasiprobability~\cite{Pashayan2015}, and non-negative representations
admit efficient simulation~\cite{MariEisert2012,Veitch2012}; see
Ref.~\cite{Ferrie2011} for a review of quasiprobability representations.
However, the relationships among these measures in the qubit setting
are incompletely understood, and in particular the connection between
purely geometric quantities and simulation cost has not been fully
characterized.

\paragraph{This work.}
We investigate the $\ell_1$ distance
\begin{equation}
  C(\rho) := \min_{W_f \in \Wfree} \|W_\rho - W_f\|_1
\end{equation}
from a state's discrete Wigner function to the convex hull of stabilizer
Wigner functions.
The discrete Wigner function is a natural phase-space
representation for this purpose: under the product construction~\cite{Wootters1987},
stabilizer states form a convex polytope $\Wfree$ with clean geometric structure,
it satisfies the discrete analogue of Hudson's theorem in odd prime
dimensions~\cite{Gross2006}, and it yields the tensor-product factorization
$W_{\rho\otimes\sigma}(\alpha,\beta) = W_\rho(\alpha)W_\sigma(\beta)$
that our multi-qubit results rely on.
By linear programming duality, computing $C$ automatically yields a
\emph{dual witness} $H^*$ that certifies the distance.

A frame must be chosen to define $C$, and for qubits no choice is
canonical: unlike the odd-prime case, where Gross's construction is the
unique Clifford-covariant one~\cite{Gross2006}, no qubit discrete Wigner
function is Clifford covariant at
all~\cite{ZhuWigner2016,Raussendorf2023cohomology}.
The mana~\cite{Veitch2014} measures Wigner negativity relative to the
non-negative orthant, which for qubits is the wrong reference set: in a
fixed frame, multi-qubit stabilizer states already have negative Wigner
entries.
Replacing the orthant by the stabilizer polytope gives a faithful
quantity, and the present work determines its geometry exactly.
We stress that the result is \emph{not} a qubit analogue of the mana in
the operational sense: the mana is a monotone under stabilizer operations,
whereas $C$ is not monotone even under a single local phase gate
(\Cref{rem:not_monotone}).
Faithfulness, convexity and the Lipschitz property of $C$ are the generic
properties of any norm-induced distance to a convex free set~\cite{Regula2018};
what is specific here is the qubit phase-space geometry.
Operationally grounded qubit quasiprobability frameworks that avoid the
frame problem altogether are the $\Lambda$-polytope models of
Refs.~\cite{Zurel2020,Okay2021,Raussendorf2020}.
We work in the Wootters product frame because product frames supply the
tensor-product factorization our multi-qubit results use, and we treat the
resulting frame dependence as an object of study rather than a nuisance:
\Cref{sec:maxC} determines which conclusions hold in every frame and which
are properties of the frame we fixed.

\paragraph{Main results.}
\begin{enumerate}[(i)]
\item \textbf{Exact single-qubit geometry.}
  The Wigner $\ell_1$ metric on single-qubit Bloch vectors is
  $N(\vec v) = \max(\|\vec v\|_\infty,\tfrac12\|\vec v\|_1)$, whose unit
  ball is a cuboctahedron (\Cref{lem:cuboct}). The distance itself has the
  closed form $C(\rho)=\tfrac12\max(0,\|\vec r\|_1-1)$
  (\Cref{prop:closed_form}), so $\max_\rho C(\rho) = (\sqrt3-1)/2$, attained
  exactly at the eight face states (\Cref{thm:maxC1}). Moreover, for
  magic states ($C>0$), $C(\rho)=\|W_\rho\|_1-1$ exactly when the Bloch
  sign product $r_xr_yr_z$ is non-positive.

\item \textbf{Products and exact tensor powers.}
  $C(\rho_1\otimes\cdots\otimes\rho_n) \leq \prod_i(1+C(\rho_i)) - 1$ for
  single-qubit factors (\Cref{prop:prod_upper}), with equality whenever
  every factor has $r_xr_yr_z\leq0$ (\Cref{thm:product_equality}; the
  equatorial two-factor case was first proved, and formalized in Lean~4, by
  Zhang and Haobo~\cite{ZhangHaobo2026}). This gives
  $C(\rho_f^{\otimes n}) = ((1+\sqrt3)/2)^n - 1$ exactly and the exact
  maximum of $C$ over fully separable $n$-qubit states
  (\Cref{cor:exact_powers}); only entangled states remain open.

\item \textbf{What is intrinsic and what is frame relative.}
  No qubit discrete Wigner function is Clifford covariant, so frame
  dependence is unavoidable. The maximum in (i) holds in both
  single-qubit frames, and the values in (ii) in every product frame, all
  of which are related by local Clifford gates. Across all $2048$
  Pauli-covariant two-qubit frames (\Cref{rem:frame_freedom}), the ratio
  $C(\rho^{Rx}_\theta)/C(\rho^{Ry}_\theta)$ takes the values $\tfrac12$,
  $1$ and $2$, and the separable maximum $\sqrt3/2$ is not attained by
  products of face states outside $128$ of them.
  The results below therefore describe the geometry of a fixed frame.

\item \textbf{Tetrahedral dichotomy.}
  In the Wootters frame the bound of (ii) is attained whenever both
  factors have non-positive Bloch sign product (proved), and in every
  sampled pair with a positive sign product it is not (numerical); the
  governing invariant is the tetrahedral sign, not the hemisphere of
  $\langle Z\rangle$. A previously conjectured factorization of the
  deficit is false (\Cref{rem:deficit_counterexample}).

\item \textbf{Exact tightness ratios with certificates.}
  $\Gamma(\rho) \geq 1 + C(\rho)/M_n$ with
  $M_n=\max_{\sigma}\|W_\sigma\|_1$, and for the $R_y$, $R_x$ and
  Bell$+R_z$ families in the two-qubit repetition codespace the tightness
  ratio $\kappa(\rho) := (\Gamma(\rho)-1)/C(\rho)$ equals exactly $1$, $2$
  and $1$ for all parameter values with $C>0$
  (\Cref{thm:kappa_ry_rx,thm:kappa_BRz}),
  each proved by an explicit optimal free state and an explicit normalized
  dual witness.
  The factor of $2$ is visible in the certificates: the optimal difference
  $W_\rho-W_{\tau^*}$ is supported on $4$ phase points for $R_x$ and on
  $8$ for $R_y$, with equal weights.

\item \textbf{Code-projected logical witnesses.}
  The optimal witnesses are logical Pauli combinations compressed to the
  codespace, $H^*=(I+ZZ)L$; for $R_y$ and Bell$+R_z$ the bare logical
  operator $L$ certifies only half of $C$ (\Cref{thm:logical_witnesses}).
  For the analogous $n=3,4$ families an optimal witness commuting with
  every code stabilizer still exists (exact linear programming).

\item \textbf{GHZ families for all $n$.}
  For $n$-qubit GHZ$+R_z$ states, $\Gamma=|\sin\phi|+|\cos\phi|$ exactly for
  every $n$, while
  $C\geq 2^{\lfloor n/2\rfloor-1}(|\sin\phi|+|\cos\phi|-1)$, with equality
  for $n=2$ and, at every angle evaluated, for $n=3,4$ (\Cref{thm:ghz}). The proof turns on a phase-parity property
  of the Wootters phase-point operators, which fixes the normalization of a
  code-projected witness at $2^{\lceil n/2\rceil}$. Hence
  $\kappa\leq2^{1-\lfloor n/2\rfloor}$.

\item \textbf{Non-monotonicity.}
  $C$ is not monotone under the Clifford group---a single local phase gate
  can halve it---so $C$ cannot replace $\Gamma$ in distillation bounds.
\end{enumerate}

\paragraph{Relation to prior work.}
The mana~\cite{Veitch2014} is built from the Wigner $\ell_1$ norm but
measures distance to the non-negative orthant rather than to the stabilizer
polytope, and provides a scalar without a geometric certificate.
The robustness of magic~\cite{Howard2017} is operationally grounded and
Clifford invariant, but lacks closed-form solutions for generic states and
is costly to compute beyond a few qubits~\cite{Heinrich2019,Hamaguchi2024}.
The quantity studied here sits differently from both: it is a distance to
the stabilizer polytope in a fixed phase-space frame, so it admits exact
extremal geometry and dual-witness certificates, at the cost of the
Clifford invariance that $\Gamma$ enjoys.
That trade-off is not a defect of the construction but a consequence of the
non-existence of Clifford-covariant qubit Wigner
functions~\cite{ZhuWigner2016}; our contribution is to compute the
resulting geometry exactly and to delimit precisely which of its features
survive the frame freedom and which do not.
Trace-distance geometry of the same polytope has recently been studied
numerically by Palhares \emph{et al.}~\cite{Palhares2026}, with which our
exact two- and three-qubit values may be compared directly.
On the frame side, the closest prior results are those of
Galv\~{a}o~\cite{Galvao2005} and Cormick \emph{et al.}~\cite{Cormick2006},
who characterize the states that are Wigner-non-negative across the whole
Gibbons--Hoffman--Wootters family; our frame analysis
(\Cref{rem:frame_freedom}) asks the complementary, quantitative question
of how a distance defined in one frame varies across frames.
For a single qubit, \Cref{lem:onequbit_nonneg} below is consistent with
their characterization: every single-qubit stabilizer mixture is
Wigner-non-negative in the Wootters frame.

\paragraph{Organization.}
\Cref{sec:background} reviews stabilizer formalism and discrete phase space.
\Cref{sec:framework} defines $C(\rho)$, $\Gamma(\rho)$, and $\kappa(\rho)$.
\Cref{sec:maxC} establishes the exact geometry of $C$: the cuboctahedral form
of the Wigner $\ell_1$ metric, the single-qubit maximum, the multiplicative
bound for products, the exact tensor powers and separable maxima, and the
enumeration of the frame freedom that separates the intrinsic results from
the frame-relative ones.
\Cref{sec:tensor} analyzes tensor products and the tetrahedral dichotomy.
\Cref{sec:kappa} contains the exact tightness-ratio results for the three
code-subspace families.
\Cref{sec:asymptotic} discusses tensor-power growth and non-monotonicity.
\Cref{sec:discussion} contains the code-projected structure of the
witnesses, the $n$-qubit GHZ theorem, and open problems.

\section{Background}
\label{sec:background}

\subsection{Stabilizer Formalism and Classical Simulability}

The $n$-qubit Pauli group $\mathcal{P}_n$ consists of $n$-fold tensor products
of $\{I, X, Y, Z\}$ with phases $\{\pm 1, \pm i\}$.
A \emph{stabilizer state} is the unique $+1$ eigenstate of an abelian
subgroup $S \leq \mathcal{P}_n$ of size $2^n$.
The set of all pure $n$-qubit stabilizer states is denoted $\Stab_n$, and
its convex hull, the set of stabilizer mixtures, $\mathrm{conv}(\Stab_n)$.

The Gottesman--Knill theorem~\cite{Gottesman1998,Aaronson2004} states that any circuit
composed of Clifford gates, computational basis measurements, and stabilizer
state preparations can be efficiently simulated classically.
The addition of a single non-stabilizer resource state enables universal
quantum computation~\cite{Bravyi2005}, making the characterization of
non-stabilizerness central.

\subsection{Magic as a Resource}

In the resource-theoretic framework~\cite{Veitch2014,Chitambar2019},
free states are mixtures of stabilizer states and free operations are
stabilizer circuits.
The key operational quantity is the \emph{robustness of
magic}~\cite{Howard2017}
\begin{multline}
  \Gamma(\rho) := \min\!\Bigl\{\textstyle\sum_i |\alpha_i| :
    \rho = \textstyle\sum_i \alpha_i \sigma_i,\\
    \sigma_i \in \Stab_n,\; \alpha_i\in\mathbb{R}\Bigr\},
  \label{eq:Gamma}
\end{multline}
the minimal $\ell_1$ weight of a real affine decomposition of $\rho$ into
pure stabilizer states; quasiprobability simulation based on such a
decomposition has classical cost scaling as $\Gamma(\rho)^2$~\cite{Howard2017}.
For pure states this quantity is related to, but distinct from, the
stabilizer extent of Ref.~\cite{Bravyi2019}.

\subsection{Discrete Phase Space and the Wigner Function}
\label{subsec:wigner}

We use the \emph{product construction} of Wootters~\cite{Wootters1987}.
The single-qubit phase-point operators are indexed by
$\alpha_k = (q_k, p_k) \in \mathbb{F}_2^2$:
\begin{equation}
  A_{(q_k,p_k)} = \tfrac{1}{2}\!\bigl(I + (-1)^{p_k} X
                  + (-1)^{q_k+p_k} Y + (-1)^{q_k} Z\bigr).
\end{equation}
For $n$ qubits, $A_\alpha = A_{\alpha_1} \otimes \cdots \otimes A_{\alpha_n}$,
and the \emph{discrete Wigner function} is
\begin{equation}
  W_\rho(\alpha) = \frac{1}{2^n}\,\tr(\rho\,A_\alpha),
  \qquad \alpha \in \mathbb{F}_2^{2n}.
  \label{eq:wigner_def}
\end{equation}
The reconstruction formula $\rho = \sum_\alpha W_\rho(\alpha) A_\alpha$
shows the Wigner transform is a bijection.
A key consequence of the product structure is \emph{Wigner factorization}:
for any product state $\rho \otimes \sigma$,
\begin{equation}
  W_{\rho\otimes\sigma}(\alpha,\beta)
  = W_\rho(\alpha)\cdot W_\sigma(\beta).
  \label{eq:wigner_factorization}
\end{equation}

Pauli operators act as phase-space translations under the product
construction: conjugation by $X$, $Y$, or $Z$ permutes the phase points
$A_\alpha$.
Unlike the odd-prime case, however, general single-qubit Cliffords do
\emph{not} act as permutations of the Wootters phase points: for example,
$S A_{(q,p)} S^\dagger$ is not a phase-point operator.
Single-qubit $C$ is nevertheless invariant under all of $\mathcal{C}_1$
(\Cref{thm:properties}(iii)), while for $n\geq 2$ even local single-qubit
Cliffords can change $C$ (\Cref{rem:cliff_equiv,rem:not_monotone}).

\begin{remark}[Why a Wigner frame]
\label{rem:wigner_choice}
Any injective linear representation of states---a Wigner frame, a
product informationally complete POVM (a discrete $Q$ function), or the Pauli
characteristic function $\chi_\rho(P)=\tr(P\rho)$---yields a faithful
distance to the image of $\mathrm{conv}(\Stab_n)$, and all three factorize
on product states (for $\chi$, $\chi_{\rho\otimes\sigma}(P\otimes
Q)=\chi_\rho(P)\chi_\sigma(Q)$).
Positivity of the representation is not an obstruction, because the free
set is a polytope rather than an orthant.
The choices differ in their geometry. The Wootters phase-point operators
are orthogonal and have unit trace, so the $\ell_1$ norm on Wigner
functions is the natural quasiprobability norm whose size governs
sampling-based simulation~\cite{Pashayan2015}, and it is this norm whose
single-qubit restriction has the cuboctahedral form of
\Cref{lem:cuboct}. The characteristic-function $\ell_1$ norm, by
contrast, underlies the stabilizer R\'enyi entropies~\cite{Leone2022} and
is Clifford covariant up to signs.
We study the Wigner choice; the others are natural comparisons.
\end{remark}

\subsection{The Subtlety of Qubits}

For odd prime $d$, Gross~\cite{Gross2006} proved the qudit Hudson theorem:
a pure state has non-negative Wigner function if and only if it is a stabilizer
state.
For qubits ($d=2$), stabilizer states can have negative Wigner
entries~\cite{Raussendorf2017,Raussendorf2020,Appleby2005}: in the Wootters
frame this first happens at $n=2$, where the Bell states have entries
$-\tfrac18$ (\Cref{eq:W_BRz} at $\theta=0$). The free region is therefore the
convex hull
\begin{equation}
  \Wfree := \mathrm{conv}\{W_\sigma : \sigma \in \Stab_n\},
\end{equation}
a polytope that for $n=1$ lies inside the non-negative orthant but for
$n\geq2$ does not.
The mana $\mathcal{M}(\rho) = \log\|W_\rho\|_1$~\cite{Veitch2014} vanishes
exactly on the non-negative orthant, whereas
$C(\rho) = \min_{f \in \Wfree}\|W_\rho - f\|_1$ vanishes exactly on
$\Wfree$; for $n\geq2$ neither set contains the other.

\section{Framework}
\label{sec:framework}

\subsection{The Free Wigner Polytope}

The stabilizer Wigner polytope is
$\Wfree := \mathrm{conv}\{W_\sigma : \sigma \in \Stab_n\} \subset \mathbb{R}^{4^n}$.
For $n=2$, the extreme points are the Wigner functions of the 60 pure
two-qubit stabilizer states.
The maximum Wigner $\ell_1$ norm over stabilizer states is
$M_n := \max_{\sigma \in \Stab_n}\|W_\sigma\|_1$;
for $n=1$, $M_1 = 1$; for $n=2$, $M_2 = 2$, attained by the Bell states
and their local Clifford equivalents; exact enumeration gives $M_3=5/2$
and $M_4=4$.
By Wigner factorization $M_{n+m}\geq M_nM_m$, so $M_n\geq2^{\lfloor n/2\rfloor}$.

\subsection{The Wigner Distance $C(\rho)$}

\begin{definition}[Wigner distance]
\label{def:C}
$C(\rho) := \min_{W_f \in \Wfree}\|W_\rho - W_f\|_1$.
\end{definition}

$C(\rho) = 0$ iff $\rho \in \mathrm{conv}(\Stab_n)$.
It is computed by a linear program whose dual has $4^n$ variables and
$|\Stab_n|$ constraints (\Cref{app:duality}); for the families studied
here we identify an optimal free state and an optimal dual witness
analytically, giving closed-form expressions valid for all parameter
values simultaneously.

\begin{theorem}[Structural properties of $C$]
\label{thm:properties}
$C$ satisfies:
\begin{enumerate}[(i)]
  \item \textbf{Faithfulness:} $C(\rho) = 0 \iff \rho \in \mathrm{conv}(\Stab_n)$.
  \item \textbf{Convexity:}
    $C(p\rho_1 + (1-p)\rho_2) \leq pC(\rho_1) + (1-p)C(\rho_2)$.
  \item \textbf{Clifford invariance (single-qubit systems):}
    for $n=1$, $C(U\rho U^\dagger) = C(\rho)$ for all $U \in \mathcal{C}_1$.
    For $n\geq 2$, $C$ is \emph{not} invariant even under local
    single-qubit Cliffords such as $S\otimes I$ and $H\otimes I$
    (\Cref{rem:cliff_equiv,rem:not_monotone}).
  \item \textbf{Lipschitz:}
    $|C(\rho) - C(\sigma)| \leq \|W_\rho - W_\sigma\|_1$.
  \item \textbf{Restricted monotonicity:}
    $C(\Phi(\rho)) \leq C(\rho)$ for every free operation $\Phi$ whose
    induced phase-space map $T_\Phi$ is stochastic, i.e.\ admits a
    nonnegative quasiprobability representation on the Wootters phase
    space.
    This class includes all Pauli channels and computational-basis
    measure-and-discard operations.
    Clifford unitaries generally lie outside this class, and $C$ is
    \emph{not} monotone under them (\Cref{rem:not_monotone}).
\end{enumerate}
\end{theorem}

Proofs are in \Cref{app:properties}. Properties (i), (ii) and (iv) hold
for any norm-induced distance to a convex free set~\cite{Regula2018}; (iii)
and (v) use the qubit phase-space structure.

\subsection{Primal--Dual Formulation and the Witness}

\begin{theorem}[Dual representation]
\label{thm:dual}
For every state $\rho$,
\begin{equation}
  C(\rho) = \max_{\|S\|_\infty \leq 1}
  \Bigl[\langle S, W_\rho\rangle
        - \max_{W_f \in \Wfree}\langle S, W_f\rangle\Bigr].
  \label{eq:dual}
\end{equation}
The optimal $S^*$ defines a \emph{dual witness}. Writing
$H_S := \sum_\alpha S(\alpha) A_\alpha$, so that
$S(\alpha)=2^{-n}\tr(H_SA_\alpha)$ and
$\langle S,W_\rho\rangle = 2^{-n}\tr(H_S\rho)$,
\begin{equation}
  C(\rho) = 2^{-n}\Bigl[\tr(H_{S^*}\rho)
            - \max_{\sigma\in\Stab_n}\tr(H_{S^*}\sigma)\Bigr],
  \label{eq:dual_operator}
\end{equation}
and any Hermitian $H$ with $|\tr(HA_\alpha)|\leq 2^n$ for all $\alpha$
gives the lower bound
$C(\rho)\geq 2^{-n}[\tr(H\rho)-\max_{\sigma\in\Stab_n}\tr(H\sigma)]$.
\end{theorem}

\begin{remark}[Normalization]
\label{rem:normalization}
The constraint $|\tr(HA_\alpha)|\leq2^n$ is essential and is not
automatic for Pauli combinations. For example, $H=Z\otimes I+X\otimes X$
has $\max_\alpha|\tr(HA_\alpha)|=2$, so $2H$ is admissible and certifies
only $C\geq\tfrac12[\tr(H\rho)-\max_\sigma\tr(H\sigma)]$; the optimal
witnesses of \Cref{sec:kappa} include stabilizer-multiplied terms that
double this.
\end{remark}

The geometry underlying \Cref{thm:dual} is the polar duality between the
stabilizer polytope and the $\Lambda$ polytope~\cite{Zurel2020,Okay2021}:
facets of $\Wfree$ correspond to vertices of $\Lambda$, which for $n=2$
were enumerated by Reichardt~\cite{Reichardt2009} and form $22{,}320$
facets in $8$ Clifford orbits~\cite{Okay2021}.
We note that the optimal witnesses $H^*$ obtained in this work are
supporting hyperplanes of $\Wfree$ and need not be facet-defining.
The same eight facet classes were recently rederived and recast as
Bell-like non-stabilizerness witnesses by Palhares
\emph{et al.}~\cite{Palhares2026}, who study the \emph{trace distance}
to the stabilizer polytope; that measure is Clifford-invariant and
monotone by construction, complementing the frame-dependent
Wigner-distance geometry studied here.

\subsection{The Robustness of Magic and the Tightness Ratio}

Recall the robustness of magic $\Gamma(\rho)$ from Eq.~\eqref{eq:Gamma}.
Note $\Gamma(\rho) \geq 1$ with equality iff $\rho \in \mathrm{conv}(\Stab_n)$.

\begin{theorem}[Simulation bound]
\label{thm:simbound}
For any $n$-qubit state $\rho$,
\begin{equation}
  \Gamma(\rho) \geq 1 + \frac{C(\rho)}{M_n}.
  \label{eq:simbound}
\end{equation}
\end{theorem}

\begin{proof}
Let $\rho = \sum_i \alpha_i \sigma_i$ achieve $\Gamma(\rho) = \sum_i |\alpha_i|$.
Setting $\lambda = (\Gamma(\rho)+1)/2$, $\mu = (\Gamma(\rho)-1)/2$, and
$\sigma^\pm$ as the normalized positive/negative parts gives
$W_\rho = \lambda W_{\sigma^+} - \mu W_{\sigma^-}$, so
\begin{equation}
  C(\rho) \leq \mu\|W_{\sigma^+} - W_{\sigma^-}\|_1
  \leq 2\mu M_n = (\Gamma(\rho)-1)M_n. \qedhere
\end{equation}
\end{proof}

\begin{definition}[Tightness ratio]
For $C(\rho) > 0$:
\begin{equation}
  \kappa(\rho) := \frac{\Gamma(\rho) - 1}{C(\rho)}.
  \label{eq:kappa}
\end{equation}
\Cref{thm:simbound} gives $\kappa(\rho) \geq 1/M_n$;
for two-qubit states, $M_2 = 2$ so $\kappa(\rho) \geq 1/2$.
\end{definition}

\begin{remark}
For stabilizer states, $C=0$ and $\Gamma=1$, so $\kappa=0/0$ is undefined.
The definition restricts to non-free states; for the families studied here
both numerator and denominator vanish only at isolated points (e.g.\
$\theta=\pi/2$ for the $R_y$ and $R_x$ families).
\end{remark}

\subsection{Classical Simulation Interpretation}

The bound~\eqref{eq:simbound} certifies simulation hardness from geometric
data alone: for $C(\rho) = c$, the overhead is at least $(1 + c/M_n)^2$.
The witness $H_{S^*}$ provides an explicit, measurable observable certifying
this cost via $\tr(H_{S^*}\rho)$ exceeding the stabilizer maximum by
exactly $2^nc$ (\Cref{eq:dual_operator}).

\section{The Exact Geometry of $C$}
\label{sec:maxC}

\subsection{The single-qubit maximum}

\begin{lemma}[Cuboctahedral form of the Wigner $\ell_1$ metric]
\label{lem:cuboct}
For single-qubit states with Bloch vectors $\vec r_\rho, \vec r_\sigma$,
\begin{gather}
  \|W_\rho - W_\sigma\|_1 = N(\vec r_\rho - \vec r_\sigma),\\
  N(\vec v) := \max\Bigl(\|\vec v\|_\infty,\ \tfrac{1}{2}\|\vec v\|_1\Bigr).
\end{gather}
The unit ball of $N$ is the cuboctahedron
$K = \{\vec v : \|\vec v\|_\infty \le 1,\ \|\vec v\|_1 \le 2\}$,
with vertices the signed permutations of $(1,1,0)$ and support function
$h_K(\vec u) = |u|_{(1)} + |u|_{(2)}$, the sum of the two largest
components of $|\vec u|$.
\end{lemma}

\begin{proof}
As in the proof of \Cref{thm:properties}(iii),
$W_\rho(q,p) = \tfrac14(1+\vec r\cdot\vec t_{qp})$, so
$\|W_\rho - W_\sigma\|_1 = \tfrac14\sum_{t\in T}|\langle\vec v, t\rangle|$
with $\vec v = \vec r_\rho - \vec r_\sigma$ and
$T = \{t\in\{\pm1\}^3 : t_xt_yt_z = +1\}$.
On the cone $x\ge y\ge z\ge 0$ the four terms are
$(x{+}y{+}z) + (x{+}y{-}z) + |x{-}y{-}z| + (x{-}y{+}z)$, which equals
$4x$ if $x\ge y+z$ and $2(x{+}y{+}z)$ otherwise, i.e.\
$4\max\bigl(\|\vec v\|_\infty, \tfrac12\|\vec v\|_1\bigr)$.
Both sides are invariant under all signed coordinate permutations
(any such map sends $T$ onto $T$ or $-T$, and the absolute values are
insensitive to this sign), so the identity extends to all of
$\mathbb{R}^3$.
The set $\{N\le1\}$ is the intersection of the cube
$\{\|\cdot\|_\infty\le1\}$ with the doubled octahedron
$\{\|\cdot\|_1\le2\}$ --- the cuboctahedron; its support function is the
maximum of $\langle\vec u, k\rangle$ over the twelve vertices, which is
$|u|_{(1)}+|u|_{(2)}$.
\end{proof}

\begin{remark}[The cuboctahedron is a root polytope]
\label{rem:root_polytope}
The traceless parts of the four phase-point operators form a regular
simplex, so $\vec v\mapsto W_\rho-W_\sigma$ is a linear bijection from
Bloch space onto the sum-zero hyperplane of $\mathbb{R}^4$, and $N$ is the
$\ell_1$ norm restricted to that hyperplane. Its unit ball is therefore the
root polytope of $A_3$, with vertices $\tfrac12(e_\alpha-e_\beta)$; the
vertex $(1,1,0)$ of $K$ maps to $\tfrac12(e_{00}-e_{01})$. The same
argument gives the $A_{d^2-1}$ root polytope for any $d\times d$ phase
space with orthogonal phase-point operators, so the metric itself is
generic. What is specific to qubits, and what the rest of this section
computes, is how this metric meets the stabilizer polytope.
\end{remark}

\begin{proposition}[Closed form for one qubit]
\label{prop:closed_form}
For every single-qubit state $\rho$, pure or mixed, with Bloch vector
$\vec r$,
\begin{equation}
  C(\rho) = \tfrac12\max\bigl(0,\ \|\vec r\|_1-1\bigr).
  \label{eq:closed_form}
\end{equation}
The Wigner function $W_\rho$ has at most one negative entry, and
\begin{equation}
  \|W_\rho\|_1-1=\tfrac12\max\Bigl(0,\ \max_{\vec u\in-T}\langle\vec r,\vec u\rangle-1\Bigr)
  \ \leq\ C(\rho),
  \label{eq:negativity_vs_C}
\end{equation}
where $-T=\{\vec u\in\{\pm1\}^3:u_xu_yu_z=-1\}$. Equality holds in
\eqref{eq:negativity_vs_C} if and only if $C(\rho)=0$ or the Bloch sign
product $s(\rho):=\operatorname{sgn}(r_xr_yr_z)$ is non-positive; for
$s(\rho)>0$ the inner maximum equals $\|\vec r\|_1-2\min_i|r_i|$.
\end{proposition}

\begin{proof}
Write $O=\mathrm{conv}\{\pm\vec e_i\}$ for the stabilizer octahedron, so
that $C(\rho)=\min_{\vec r'\in O}N(\vec r-\vec r')$ by \Cref{lem:cuboct}.
\emph{Lower bound.} $N\geq\tfrac12\|\cdot\|_1$ and $\|\vec r'\|_1\leq1$ on
$O$ give $C(\rho)\geq\tfrac12(\|\vec r\|_1-1)$.
\emph{Upper bound.} If $\|\vec r\|_1\leq1$ then $\vec r\in O$. Otherwise let
$e=\|\vec r\|_1-1>0$ and order $|r_1|\geq|r_2|\geq|r_3|$. Then
$\sum_i\min(|r_i|,e/2)\geq e$: if $|r_2|\geq e/2$ the first two terms give
$e$; otherwise $|r_1|\geq(1+e)/3>e/2>|r_2|$ and the sum exceeds $e$ by
$\tfrac12(1-|r_1|+|r_2|+|r_3|)\geq0$. Choose $a_i\in[0,\min(|r_i|,e/2)]$ with $\sum_ia_i=e$ and
put $v_i=\operatorname{sgn}(r_i)\,a_i$. Then $\vec r'=\vec r-\vec v$ has the
signs of $\vec r$ and $\|\vec r'\|_1=1$, so $\vec r'\in O$, while
$\|\vec v\|_\infty\leq e/2=\tfrac12\|\vec v\|_1$ gives $N(\vec v)=e/2$.
\emph{Wigner function.} $W_\rho(q,p)=\tfrac14(1+\langle\vec r,\vec t_{qp}\rangle)$
with $\vec t_{qp}\in T$. For distinct $\vec t,\vec t'\in T$,
$(\vec t+\vec t')/2=\pm\vec e_i$ for some $i$, so the two entries sum to
$\tfrac12(1\pm r_i)\geq0$ and at most one entry is negative. Hence
$\|W_\rho\|_1-1=2\max(0,-\min_{\vec t\in T}W_\rho)$, which is
\eqref{eq:negativity_vs_C}. Finally
$\max_{\vec u\in-T}\langle\vec r,\vec u\rangle\leq\|\vec r\|_1$, with
equality exactly when a sign vector of $\vec r$ (free on zero components)
has product $-1$, i.e.\ when some $r_i=0$ or $r_xr_yr_z<0$. If
$r_xr_yr_z>0$ the best $\vec u\in-T$ flips the sign of the smallest
component, giving $\|\vec r\|_1-2\min_i|r_i|<\|\vec r\|_1$, and the two sides
of \eqref{eq:negativity_vs_C} differ whenever $\|\vec r\|_1>1$.
\end{proof}

\begin{theorem}[Single-qubit maximum of $C$]
\label{thm:maxC1}
Let $c_0 := (\sqrt{3}-1)/2$.
Over all single-qubit states,
\begin{equation}
  \max_\rho C(\rho) = c_0,
\end{equation}
attained exactly at the eight face states, i.e.\ the pure states with
Bloch vectors $(\pm1,\pm1,\pm1)/\sqrt{3}$.
\end{theorem}

\begin{proof}
By \eqref{eq:closed_form} and $\|\vec r\|_1\leq\sqrt3\,\|\vec r\|_2\leq\sqrt3$,
$C(\rho)\leq c_0$, with equality if and only if $\|\vec r\|_2=1$ and
$|r_x|=|r_y|=|r_z|$, i.e.\ at the face states.
\end{proof}

\begin{remark}
The face states are the $T$-type magic states, long known to be the
maximally magic single-qubit states---they maximise the robustness of
magic, with $\Gamma(\rho_f)=\sqrt3$~\cite{Howard2017,Heinrich2019}, and
as the two qubit SICs ($T$ and $-T$) they maximise the stabilizer
R\'enyi entropies~\cite{CuffaroFuchs2024}---so
their appearance here is expected rather than surprising, and different
monotones need not agree on maximisers in general~\cite{Veitch2014}.
What \Cref{prop:closed_form} adds is the exact value of the Wigner distance
everywhere on the Bloch ball, and the exact relation
\eqref{eq:negativity_vs_C} between $C$ and the Wigner negativity, which is
what decides the product structure of \Cref{sec:tensor}.
Note also $\kappa(\rho_f) = (\Gamma-1)/C = (\sqrt3-1)/c_0 = 2$.
\end{remark}

\subsection{Products of single-qubit states}

\begin{lemma}[Wigner factorization]
$W_{\rho\otimes\sigma}(\alpha,\beta) = W_\rho(\alpha)\cdot W_\sigma(\beta)$.
\end{lemma}
\begin{proof}
Immediate from $A_{(\alpha,\beta)}=A_\alpha\otimes A_\beta$.
\end{proof}

\begin{lemma}[Single-qubit free states are nonnegative]
\label{lem:onequbit_nonneg}
Every pure single-qubit stabilizer state has Wigner entries in
$\{0,\tfrac12\}$; consequently every $\tau\in\Wfree^{(1)}$ satisfies
$W_\tau\geq0$ and $\|W_\tau\|_1 = 1$.
(Wigner negativity of stabilizer states first occurs at $n=2$,
cf.\ $M_1=1$ versus $M_2=2$.)
\end{lemma}
\begin{proof}
For $\vec r = \pm\vec e_i$ the entries
$\tfrac14(1+\vec r\cdot\vec t_{qp})$ equal $\tfrac14(1\pm1)$, two of
each; mixtures inherit nonnegativity, and then
$\|W_\tau\|_1 = \sum_\alpha W_\tau(\alpha) = 1$.
\end{proof}

\begin{proposition}[Multiplicative upper bound for products]
\label{prop:prod_upper}
For any single-qubit states $\rho_1,\dots,\rho_n$,
\begin{equation}
  C(\rho_1\otimes\cdots\otimes\rho_n)
  \;\leq\; \prod_{i=1}^n\bigl(1+C(\rho_i)\bigr) - 1 .
\end{equation}
\end{proposition}
\begin{proof}
Let $\tau_i\in\Wfree^{(1)}$ attain $C(\rho_i)$ and write
$W_{\rho_i} = W_{\tau_i} + \Delta_i$ with
$\|\Delta_i\|_1 = C(\rho_i)$.
By the Wigner factorization lemma,
$W_{\rho_1\otimes\cdots\otimes\rho_n}
 = \sum_{T\subseteq[n]} \bigotimes_i X_i^{(T)}$
with $X_i^{(T)} = \Delta_i$ for $i\in T$ and $X_i^{(T)} = W_{\tau_i}$
otherwise.
The $T=\varnothing$ term is the Wigner function of the free state
$\tau_1\otimes\cdots\otimes\tau_n$, so by the triangle inequality and
the multiplicativity of $\ell_1$ norms under tensor products,
\begin{align*}
  C(\rho_1\otimes\cdots\otimes\rho_n)
  &\leq \sum_{\varnothing\neq T\subseteq[n]}
        \prod_{i\in T}\|\Delta_i\|_1 \prod_{i\notin T}\|W_{\tau_i}\|_1\\
  &= \prod_{i=1}^n\bigl(1+C(\rho_i)\bigr) - 1,
\end{align*}
using $\|W_{\tau_i}\|_1 = 1$ from \Cref{lem:onequbit_nonneg}.
\end{proof}

The matching lower bound holds whenever every factor has non-positive
Bloch sign product.

\begin{theorem}[Product equality on the non-positive branch]
\label{thm:product_equality}
For any single-qubit states $\rho_1,\dots,\rho_k$,
\begin{equation}
  C(\rho_1\otimes\cdots\otimes\rho_k)\ \geq\ \prod_{i=1}^k\|W_{\rho_i}\|_1-1 .
  \label{eq:product_lower}
\end{equation}
Consequently, if every factor satisfies $s(\rho_i)\leq0$ or $C(\rho_i)=0$,
\begin{equation}
  C(\rho_1\otimes\cdots\otimes\rho_k)=\prod_{i=1}^k\bigl(1+C(\rho_i)\bigr)-1 .
  \label{eq:product_equality}
\end{equation}
\end{theorem}

\begin{proof}
For each $i$ let $s_i\in\{\pm1\}^4$ equal $-1$ at the negative entry of
$W_{\rho_i}$, if there is one (\Cref{prop:closed_form}), and $+1$
elsewhere, so that $\langle s_i,W_{\rho_i}\rangle=\|W_{\rho_i}\|_1$. Put
$S=s_1\otimes\cdots\otimes s_k$; then $\|S\|_\infty=1$ and, by Wigner
factorization, $\langle S,W_{\rho_1\otimes\cdots\otimes\rho_k}\rangle
=\prod_i\|W_{\rho_i}\|_1$. The corresponding operator is
$H_S=H_1\otimes\cdots\otimes H_k$ with $H_i=\sum_\alpha s_i(\alpha)A_\alpha$,
which equals $2I$ if $s_i\equiv1$ and $2I-2A_{\alpha_i}=I-\vec t_{\alpha_i}\cdot\vec\sigma$
otherwise; the latter has all four Pauli coefficients equal to $\pm1$. Let
$\mathcal A$ be the set of factors of the second kind. For a pure
stabilizer state $\sigma$ with stabilizer group $\mathcal S$,
$2^{-k}\tr(\sigma H_S)=2^{-|\mathcal A|}\tr\bigl(\sigma_{\mathcal A}\bigotimes_{i\in\mathcal A}H_i\bigr)$,
where $\sigma_{\mathcal A}$ is the reduced state. Every Pauli coefficient of
$\bigotimes_{i\in\mathcal A}H_i$ is $\pm1$, and $\tr(\sigma_{\mathcal A}P)\neq0$
only if $\pm P\otimes I\in\mathcal S$; such $P$ form an abelian group of
$|\mathcal A|$-qubit Paulis not containing $-I$, of order at most
$2^{|\mathcal A|}$. Hence $2^{-k}\tr(\sigma H_S)\leq1$, so
$\langle S,W_f\rangle\leq1$ on $\Wfree$ by convexity, and \Cref{thm:dual}
gives \eqref{eq:product_lower}. If $s(\rho_i)\leq0$, then
$\|W_{\rho_i}\|_1=1+C(\rho_i)$ by \Cref{prop:closed_form}; if $C(\rho_i)=0$,
then $\|W_{\rho_i}\|_1=1$ by \Cref{lem:onequbit_nonneg}. In either case
\eqref{eq:product_lower} matches the upper bound of \Cref{prop:prod_upper}.
\end{proof}

\begin{remark}[Attribution]
\label{rem:attribution}
For two equatorial factors ($r_z=0$, hence $s=0$), \eqref{eq:product_equality}
is the theorem of Zhang and Haobo~\cite{ZhangHaobo2026}, who proved it with
the product sign witness used above and formalized it in Lean~4; their
development also proves $C(\rho\otimes\rho)\geq2C(\rho)$ on the branch
$r_xr_yr_z\leq0$. The extension to all factors with $s\leq0$ and to $k$
factors combines their witness with \Cref{prop:closed_form}. For a factor
with $s>0$ and $C>0$ the witness bound \eqref{eq:product_lower} lies
strictly below the product formula, because $\|W_\rho\|_1-1<C(\rho)$
there.
\end{remark}

\subsection{Two qubits: the product of two faces}

\begin{proposition}[Exact value at the product of two faces]
\label{prop:maxC2}
In the Wootters product frame of \Cref{subsec:wigner}, let $\rho_f$ be the
face state with Bloch vector $(1,1,-1)/\sqrt3$.
Then
\begin{equation}
  C(\rho_f\otimes\rho_f) = \frac{\sqrt3}{2}.
\end{equation}
Equivalently, and in every product frame,
$\sqrt3/2$ is the maximum of $C$ over separable two-qubit states
(\Cref{cor:exact_powers}); the particular face attaining it depends on the
frame (\Cref{rem:frame_freedom}).
\end{proposition}

\begin{proof}
\emph{Lower bound (witness).}
Let $S_1(\alpha)\in\{\pm1\}$ equal $+1$ at every single-qubit phase
point except the one where $W_{\rho_f}$ is negative
(the point $(q,p)=(0,1)$, with $\vec t_{(0,1)} = (-1,-1,1)$
antiparallel to $\vec f$), and set
$S(\alpha,\beta) = S_1(\alpha)S_1(\beta)$, so $\|S\|_\infty = 1$.
The corresponding operator is $H\otimes H$ with
\[
  H = \sum_\alpha S_1(\alpha)A_\alpha = 2I - 2A_{(0,1)} = I + X + Y - Z,
\]
and $\langle S, W_\rho\rangle = \tfrac14\tr(\rho\,H{\otimes}H)$.
Every coefficient of $H\otimes H$ in the Pauli basis is $\pm1$.
A pure two-qubit stabilizer state is
$\sigma = \tfrac14\sum_{g\in\mathcal S}g$ over its signed stabilizer
group of four elements, so
\[
  \tfrac14\tr(\sigma\,H{\otimes}H)
  = \tfrac14\sum_{g\in\mathcal S}(\pm1) \ \le\ 1
\]
for \emph{every} pure stabilizer state, entangled or not.
On the target state, Wigner factorization gives
$\tfrac14\tr(\rho_f^{\otimes2}H^{\otimes2})
 = \bigl(\tfrac12\tr(\rho_f H)\bigr)^2
 = \bigl(\tfrac{1+\sqrt3}{2}\bigr)^2 = \tfrac{2+\sqrt3}{2}$.
By \Cref{thm:dual},
$C(\rho_f\otimes\rho_f) \ge \tfrac{2+\sqrt3}{2} - 1 = \tfrac{\sqrt3}{2}$.

\emph{Upper bound (explicit free state).}
Consider the stabilizer mixture
\begin{align*}
\tau^* ={}& \tfrac{4-\sqrt3}{12}\,\ketbra{1}{1}\otimes\ketbra{+}{+}
       \;+\; \tfrac{4-\sqrt3}{12}\,\ketbra{1}{1}\otimes\ketbra{{+}i}{{+}i} \\
      &+ \tfrac{\sqrt3}{4}\,\ketbra{+}{+}\otimes\ketbra{1}{1}
       \;+\; \tfrac{5\sqrt3-8}{12}\,\ketbra{{+}i}{{+}i}\otimes\ketbra{1}{1} \\
      &+ \tfrac{2-\sqrt3}{2}\cdot
         \tfrac14\bigl(I{\otimes}I + X{\otimes}X - Y{\otimes}Z - Z{\otimes}Y\bigr),
\end{align*}
where the last term is the entangled stabilizer state with stabilizer
group $\{II,\, XX,\, -YZ,\, -ZY\}$.
All weights are positive and sum to one, and a direct computation,
exact in $\mathbb{Q}(\sqrt3)$, gives
$\|W_{\rho_f\otimes\rho_f} - W_{\tau^*}\|_1 = \sqrt3/2$.
\end{proof}

\begin{remark}
\Cref{prop:maxC2} saturates the product bound of \Cref{prop:prod_upper} at
$\rho=\sigma=\rho_f$: $\sqrt3/2 = 2c_0 + c_0^2$ with
$c_0 = C(\rho_f) = (\sqrt3-1)/2$.
This is an instance of the equality branch of the tetrahedral dichotomy
established in \Cref{sec:tensor} (both factors have sign product $s=-1$).
Note also that the optimal free state $\tau^*$ contains an entangled
stabilizer component even though superadditivity holds with equality
here.
In the trace-distance geometry of Ref.~\cite{Palhares2026}, random
sampling independently identifies the product of two face states as the
leading candidate for the two-qubit maximum (found numerically and
explicitly noted there to be separable), and that state maximizes no
single facet inequality --- the trace-distance analogue of our
observation that the optimal witnesses are supporting hyperplanes
rather than facets.
\Cref{prop:maxC2} proves the corresponding exact value in the Wigner
metric.
\end{remark}

\begin{corollary}[Unconditional exponential growth along face states]
\label{cor:maxCn}
For every $n\ge1$,
\begin{equation}
  C\bigl(\rho_f^{\otimes n}\bigr)
  \ \ge\ \Bigl(\tfrac{1+\sqrt3}{2}\Bigr)^{n} - 1 .
\end{equation}
\end{corollary}
\begin{proof}
The witness $S_1^{\otimes n}$ (operator $H^{\otimes n}$) has all Pauli
coefficients $\pm1$, and an $n$-qubit pure stabilizer group has $2^n$
elements, so $2^{-n}\tr(\sigma H^{\otimes n}) \le 1$ for every pure
stabilizer state---equivalently, H\"{o}lder's inequality against the
stabilizer norm of Ref.~\cite{Campbell2011}, which equals $1$ on pure
stabilizer states---while
$2^{-n}\tr(\rho_f^{\otimes n}H^{\otimes n})
 = \bigl(\tfrac{1+\sqrt3}{2}\bigr)^n$.
For the analogous tensor-power behaviour of the robustness of magic,
obtained through the symmetries of the stabilizer polytope, see
Ref.~\cite{Heinrich2019}.
\end{proof}

\begin{corollary}[Exact tensor powers and the separable maximum]
\label{cor:exact_powers}
For every $n\geq1$, the maximum of $C$ over fully separable $n$-qubit
states equals
\begin{equation}
  \max_{\rho\ \mathrm{separable}} C(\rho)
  = \Bigl(\tfrac{1+\sqrt3}{2}\Bigr)^{\!n} - 1 :
\end{equation}
by convexity of $C$ the maximum over the separable set is attained at pure
product states, where \Cref{prop:prod_upper} and \Cref{thm:maxC1} give the
bound.
It is attained by $\rho_f^{\otimes n}$, with $\rho_f$ the face state of
\Cref{prop:maxC2}, for which
$C(\rho_f^{\otimes n}) = ((1+\sqrt3)/2)^n - 1$ exactly---the lower bound
being \Cref{cor:maxCn} and the upper bound \Cref{prop:prod_upper} with
$C(\rho_f) = c_0$.
The value of the maximum is the same in every product frame, since these
are related by local Clifford gates (\Cref{rem:frame_freedom}); which of the
eight faces attains it is not.
In particular the separable two-qubit maximum is $\sqrt3/2$, so only
entangled states remain in question for the global two-qubit maximum.
\end{corollary}

\subsection{The frame freedom}

\begin{remark}[The frame freedom, and what survives it]
\label{rem:frame_freedom}
No Clifford-covariant discrete Wigner function exists for
qubits~\cite{ZhuWigner2016,Raussendorf2023cohomology}, so any frame-based
measure on qubits is necessarily non-covariant; the question is not whether
$C$ depends on the frame but which of its features do.
For a single qubit the frame freedom can be settled completely.
Antonopoulos \emph{et al.}~\cite{Antonopoulos2025} characterise all valid
discrete Wigner functions on a $d\times d$ phase space by four conditions on
the phase-point operators---Hermiticity, unit trace, orthogonality
$\tr[A(\alpha)^\dagger A(\beta)] = d\,\delta_{\alpha\beta}$, and covariance
under the Weyl--Heisenberg displacements, which for $d=2$ are the
single-qubit Paulis.
Writing $A(0) = \tfrac12(I+\vec a\cdot\vec\sigma)$, orthogonality at
$\beta=\alpha$ forces $\|\vec a\|_2^2 = 3$, while orthogonality against the
$X$-translate gives $a_x^2-a_y^2-a_z^2 = -1$; together these yield
$a_x^2=a_y^2=a_z^2=1$.
The valid single-qubit origin operators are therefore exactly the eight
sign choices $\vec a = (\pm1,\pm1,\pm1)$; translating the origin permutes
them in two orbits of four, so there are two frames up to translation, the
tetrahedra $T$ and $-T$. $C$ takes the \emph{same} value on every
state in both, because the
metric $N$ of \Cref{lem:cuboct} depends only on $|\langle\cdot,t\rangle|$.
Consequently \Cref{thm:maxC1} is a property of the state alone, not of the
representation.

For $n$ qubits the same argument applies with the $n$-qubit Paulis as
displacements. Writing $A(0)=2^{-n}\bigl(I+\sum_{P\neq I}c_P P\bigr)$,
translation acts as $c_P\mapsto\chi_\beta(P)\,c_P$ with
$\chi_\beta(P)=\pm1$ the commutation sign, so orthogonality to all
translates reads $\sum_{P}\chi_\beta(P)\,c_P^2 = 4^n\delta_{\beta,0}$ (with
$c_I=1$); inverting this Walsh--Hadamard transform gives $c_P^2=1$ for
every $P$. The valid origin operators are therefore exactly the
$2^{4^n-1}$ sign patterns $c_P=\pm1$, and since translating the origin
permutes them in orbits of $4^n$, they define $2^{4^n-1}/4^n$ distinct
frames: $2$ for $n=1$ (the tetrahedra $T$ and $-T$) and $2048$ for $n=2$.
Of the latter, $304$ arise from choosing a partition of the $15$
non-identity Paulis into five commuting triples and one joint eigenstate
$\ket{\psi_k}$ per triple, $A(0)=\sum_k\ketbra{\psi_k}{\psi_k}-I$; the
product frame is among them.
The finiteness of this family is specific to qubits: every non-identity
Pauli is its own inverse, so Hermiticity leaves only the signs $c_P=\pm1$,
whereas for $d>2$ displacements pair with distinct inverses and the
coefficients retain continuous phases.

No frame in this family makes every stabilizer state Wigner-nonnegative.
For a stabilizer state with group $\{I,P_1,P_2,P_3\}$, $P_1P_2=\lambda P_3$,
the translates realise each pair of commutation signs on $(P_1,P_2)$
equally often, and a direct count gives $\|W_\sigma\|_1=1$ if
$\lambda\,c_{P_1}c_{P_2}c_{P_3}=+1$ and $\|W_\sigma\|_1=2$ otherwise.
Requiring the first case for all fifteen commuting triples is impossible:
on the six lines of the Mermin--Peres square~\cite{Mermin1990} each of the
nine Paulis appears twice, so the product of the six conditions forces
$\prod\lambda=+1$, whereas $\prod\lambda=-1$ there.
Hence $M_2=2$ in every Pauli-covariant two-qubit frame, and the bound
$\kappa\geq1/M_2=\tfrac12$ of \Cref{thm:simbound} is frame independent;
this is the phase-space form of the obstruction discussed in
Ref.~\cite{Raussendorf2020}.

Two mechanisms organise what survives.
First, for any Clifford $U$ and frame $F$,
$C_{UFU^\dagger}(U\rho U^\dagger)=C_F(\rho)$, because $U$ permutes the
stabilizer states. The product frames are the $2^n$ chirality choices
$(\pm T)\otimes\cdots\otimes(\pm T)$, and the phase gate on qubit $j$
reverses the chirality of factor $j$ (\Cref{rem:dichotomy_origin}) while
preserving the set of fully separable states; hence
\Cref{cor:exact_powers} holds in every product frame. The same identity,
with $\rho^{Rx}_\theta=(S^\dagger\otimes I)\,\rho^{Ry}_\theta\,(S\otimes I)$,
gives $C_{T\otimes T}(\rho^{Rx}_\theta)=C_{(-T)\otimes T}(\rho^{Ry}_\theta)$:
the value $\tfrac12$ of $C(\rho^{Rx}_\theta)/C(\rho^{Ry}_\theta)$ in
\Cref{thm:kappa_ry_rx} is inverted to $2$ already in a product frame of
mixed chirality.
Second, beyond product frames we evaluate $C$ by linear programming in all
$2048$ frames. At $\theta=\pi/3$ the pair
$\bigl(C(\rho^{Rx}_\theta),C(\rho^{Ry}_\theta)\bigr)$ takes each of the
four values $\{c_0/2,\,c_0\}^2$ on exactly $512$ frames, so the ratio takes
the values $\tfrac12$, $1$ and $2$ (on the $304$ triple-partition frames the
split is $80,80,80,64$, the last being the class with both values
$c_0/2$). The same three ratios, on classes of the same sizes, occur at
every other $\theta\in(0,\pi)$ we evaluated (seven values); at
$\theta=\pi/2$ both states are stabilizer states.
We find no evidence that the separable maximum survives: products of face states
attain $\sqrt3/2$ in only $128$ of the $2048$ frames, and multistart
optimisation over pure product states in frames sampled from each of the
remaining classes finds no value above $0.74$.
The single-qubit maximum is thus intrinsic, the separable maxima are
intrinsic to product frames, and the $\kappa$ trichotomy, the tetrahedral
dichotomy of \Cref{thm:dichotomy,obs:dichotomy} and the noise asymmetry of
\Cref{subsec:noise} are properties of the pair (state, frame).
We note that the classification of Ref.~\cite{Antonopoulos2025} is for a
single $d$-dimensional system, whose displacement group for $d=4$ is
$\mathbb{Z}_4$ Weyl--Heisenberg rather than the two-qubit Pauli group; the
multi-qubit frames above are covariant under the latter and lie outside
that classification.

This is the precise sense in which the Wigner distance is a frame-relative
diagnostic: the extremal geometry it assigns to the stabilizer polytope is
intrinsic for a single qubit and across product frames, while the finer structure it resolves within a fixed frame---the
content of \Cref{sec:kappa} and \Cref{sec:tensor}---is meaningful relative
to that frame, as is any quantity a Clifford-invariant measure such as
$\Gamma$ necessarily averages away.
\end{remark}

\begin{remark}[The opposite tetrahedron]
\label{rem:plus_faces}
Within the Wootters product frame, the witness construction applies
precisely to the four faces in the local-Pauli orbit of $(1,1,-1)/\sqrt3$
(tetrahedral sign $s=-1$).
For the $s=+1$ faces, e.g.\ $\rho_g$ with Bloch vector
$(1,1,1)/\sqrt3$, the candidate operator
$I+X+Y+Z = 2A_{(0,0)}$ has Wigner coefficients $(2,0,0,0)$ and is
dual-infeasible, and the tensor powers behave very differently:
numerically $C(\rho_g^{\otimes2}) \doteq \sqrt3/4$ and
$C(\rho_g^{\otimes3}) \approx 0.5245$, far below the $s=-1$ values.
The two orbits therefore behave quite differently under tensoring; this is
the phenomenon analysed as the tetrahedral dichotomy in
\Cref{sec:tensor}.
We stress that the labelling of which tetrahedron is the ``witness-feasible''
one is fixed by the frame, not by the states: a different admissible choice
of phase-point origin exchanges the two orbits
(\Cref{rem:frame_freedom}).
What is the same in every product frame is that the eight faces split into
two orbits of four with the above pair of behaviours.
\end{remark}

\section{Tensor Products and the Tetrahedral Dichotomy}
\label{sec:tensor}

\subsection{Setup}

For odd prime dimensions every stabilizer Wigner function is a probability
distribution, so the analogous distance obeys $C_d(\rho)\geq\|W_\rho\|_1-1$;
the inequality can be strict, since Wigner-nonnegative states outside the
stabilizer polytope exist~\cite{Veitch2012}.
For two qubits the free polytope $\Wfree^{(2)}$ contains entangled
stabilizer states with negative Wigner entries, so this argument is
unavailable; \Cref{thm:product_equality} replaces it for product states by
the product bound $C\geq\prod_i\|W_{\rho_i}\|_1-1$. Whether that bound is
tight is decided by the single-qubit relation \eqref{eq:negativity_vs_C}
between $C$ and the Wigner negativity, and hence by the sign product
$s(\rho)=\operatorname{sgn}(r_xr_yr_z)$.

\subsection{The Tetrahedral Dichotomy}

Let $\rho_H$ denote the equatorial $H$-type magic state with Bloch vector
$(1/\sqrt{2},1/\sqrt{2},0)$ (the state $T\ket{+}$); it has
$C(\rho_H)=(\sqrt2-1)/2\approx 0.207$.

\begin{theorem}[Tetrahedral dichotomy, equality branch]
\label{thm:dichotomy}
For single-qubit states $\rho,\sigma$ with $s(\rho)\leq0$ and
$s(\sigma)\leq0$,
\begin{equation}
  C(\rho \otimes \sigma) = C(\rho) + C(\sigma) + C(\rho)C(\sigma),
\end{equation}
i.e.\ the upper bound of \Cref{prop:prod_upper} is attained; the same holds
for any number of factors (\Cref{thm:product_equality}).
\end{theorem}

\begin{observation}[Tetrahedral dichotomy, strict branch]
\label{obs:dichotomy}
In the Wootters frame, for single-qubit states $\rho,\sigma$ with
$C(\rho),C(\sigma)>0$ and $s(\rho)>0$ or $s(\sigma)>0$, every pair we
sampled satisfies
\begin{equation}
  C(\rho \otimes \sigma) < C(\rho) + C(\sigma) + C(\rho)C(\sigma).
\end{equation}
\end{observation}

\Cref{thm:dichotomy} is the case $k=2$ of \Cref{thm:product_equality}.
For the strict branch the product witness gives only
$C(\rho\otimes\sigma)\geq\|W_\rho\|_1\|W_\sigma\|_1-1$, which lies strictly
below the product formula there (\Cref{rem:attribution}), so the
observation is not implied. Numerically the separation is perfect: among
$300$ Haar-random pure pairs with $C(\rho),C(\sigma)>0$, all $78$ pairs with
$s\leq0$ on both factors have zero deficit, as the theorem requires, and
all $222$ other pairs have deficit at least $0.006$ (\Cref{fig:dichotomy}a).
Representative exact points: $C(\rho_f\otimes\rho_f) = \sqrt3/2$
(\Cref{prop:maxC2}), while
$C(\rho_g\otimes\rho_g) = 0.4330127\ldots \doteq \sqrt3/4$ and
$C(\rho_g\otimes\rho_f) = C(\rho_f\otimes\rho_g) = 0.6160254\ldots$,
where $\rho_g$ and $\rho_f$ denote the face states with Bloch vectors
$(1,1,1)/\sqrt3$ and $(1,1,-1)/\sqrt3$.
Unlike any hemispheric condition, the criterion is invariant under
local Pauli conjugations and under qubit swap, as any such law must be
(\Cref{rem:pauli_covariance}).

\begin{remark}[Hemispheres cannot govern the dichotomy]
\label{rem:pauli_covariance}
Local Pauli conjugations act as phase-space translations, so
$C(\rho\otimes\sigma)$, $C(\rho)$, and $C(\sigma)$ are all invariant
under $\sigma \mapsto P\sigma P^\dagger$.
But $\langle Z\rangle_\sigma$ is not: conjugating the second factor by
$X$ flips $\langle Z\rangle_\sigma \mapsto -\langle Z\rangle_\sigma$
while leaving every $C$-value unchanged, so any law conditioned on the
hemisphere of $\sigma$ is inconsistent with Pauli covariance --- each
``northern'' deficit pair has a ``southern'' image with an identical
deficit.
The Pauli invariants of a Bloch vector are $(|r_x|,|r_y|,|r_z|)$
together with the sign product $s$; the dichotomy is governed by $s$.
An earlier version of this work stated the dichotomy in terms of the
hemisphere of $\langle Z\rangle_\sigma$; that version is correct on
meridian slices where $\operatorname{sgn}\langle Z\rangle_\sigma$
coincides with $s(\sigma)$, as in \Cref{tab:dichotomy}, but fails in
general: e.g.\ $C(\rho_g\otimes\rho_f) \approx 0.6160 \neq \sqrt3/2$
even though $\langle Z\rangle_{\rho_f} < 0$.
\end{remark}

Representative data is in \Cref{tab:dichotomy}.

\begin{table}[ht]
\centering
\caption{Saturation (SA) of the product bound of \Cref{prop:prod_upper}
for $\rho_H\otimes\sigma$ at fixed azimuthal angle $\phi=1.05$, with
$C(\rho_H)\approx 0.207$; $\delta$ is the deficit.
On this meridian slice $\operatorname{sgn}\langle Z\rangle_\sigma$
coincides with the tetrahedral sign $s(\sigma)$
(\Cref{rem:pauli_covariance}).}
\label{tab:dichotomy}
\setlength{\tabcolsep}{3pt}
\small
\begin{tabular}{@{}lrrc@{}}
\toprule
$\sigma$ & $\langle Z\rangle_\sigma$ & $C(\rho_H\!\otimes\!\sigma)$ & SA? \\
\midrule
Equatorial ($\theta\!=\!\pi/2$) & $\phantom{+}0.000$ & $0.427$ & \checkmark \\
Northern  ($\theta\!=\!1.20$)   & $+0.362$           & $0.454$ & $\times$ ($\delta\!=\!0.136$) \\
Northern  ($\theta\!=\!0.35$)   & $+0.939$           & $0.332$ & $\times$ ($\delta\!=\!0.121$) \\
Southern  ($\theta\!=\!1.94$)   & $-0.361$           & $0.590$ & \checkmark \\
Southern  ($\theta\!=\!2.79$)   & $-0.939$           & $0.454$ & \checkmark \\
\bottomrule
\end{tabular}
\end{table}

\begin{figure*}[ht]
  \centering
  \includegraphics[width=0.9\textwidth]{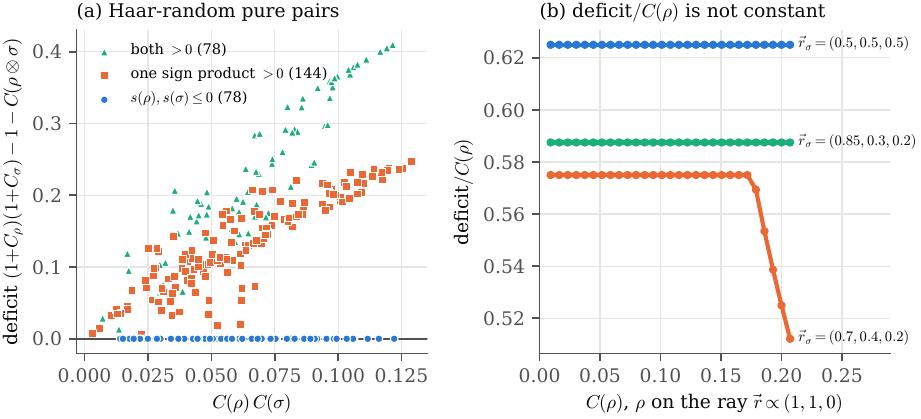}
  \caption{The tetrahedral dichotomy in the Wootters frame.
  (a)~Deficit $(1+C(\rho))(1+C(\sigma))-1-C(\rho\otimes\sigma)$ against
  $C(\rho)C(\sigma)$ for $300$ Haar-random pure single-qubit pairs with
  $C(\rho),C(\sigma)>0$, marked by how many of the two sign products
  $s=\operatorname{sgn}(r_xr_yr_z)$ are positive. Pairs with
  $s(\rho),s(\sigma)\leq0$ (blue circles) have zero deficit, as proved in
  \Cref{thm:dichotomy}; every other pair has a strictly positive deficit
  (\Cref{obs:dichotomy}).
  (b)~Deficit divided by $C(\rho)$ for equatorial $\rho$ on the ray
  $\vec r\propto(1,1,0)$ and three fixed $\sigma$ with $s(\sigma)>0$. The ratio is
  constant for two of them but drops for $\vec r_\sigma=(0.7,0.4,0.2)$, so no
  factorization $C(\rho)f(\sigma)$ holds (\Cref{rem:deficit_counterexample}).}
  \label{fig:dichotomy}
\end{figure*}

\subsection{Mechanism}

\Cref{prop:closed_form,thm:product_equality} explain the dichotomy
exactly. A single-qubit magic state ($C>0$) has $C(\rho)=\|W_\rho\|_1-1$
precisely when its sign vector can be matched by a vertex of the dual tetrahedron $-T$,
i.e.\ when $s(\rho)\leq0$; for such states all of the distance is carried
by the single negative Wigner entry, and the product sign witness of
\Cref{thm:product_equality} certifies the full product formula. For
$s(\rho)>0$ the sign vector of $\vec r$ points into the Wootters
tetrahedron $T$ itself, the negativity falls short of the distance by
$\min(\min_i|r_i|,\tfrac12(\|\vec r\|_1-1))$ whenever $C(\rho)>0$, and the product witness is no
longer tight; numerically, joint stabilizer mixtures then beat every
product of single-qubit optima.
The asymmetry is an artifact of the Wootters tetrahedral orientation and
is not intrinsic (\Cref{rem:dichotomy_origin}).

\subsection{Consequences and the Deficit}

\begin{corollary}[Equatorial multiplicativity~\cite{ZhangHaobo2026}]
\label{conj:equatorial}
For $\langle Z\rangle_\rho = \langle Z\rangle_\sigma = 0$:
$C(\rho \otimes \sigma) = C(\rho) + C(\sigma) + C(\rho)C(\sigma)$.
\end{corollary}

\begin{proof}
$r_z=0$ gives $s=0$ for both factors; apply \Cref{thm:dichotomy}.
\end{proof}

\begin{corollary}[Self-tensor superadditivity, $s\leq0$ branch]
\label{conj:selftensor}
For any qubit state $\rho$ with $s(\rho)\leq 0$:
$C(\rho^{\otimes k}) = (1+C(\rho))^k-1$; in particular
$C(\rho \otimes \rho) \geq 2C(\rho)$.
\end{corollary}

\begin{proof}
\Cref{thm:product_equality} and $(1+C)^2-1\geq2C$.
\end{proof}

Without the restriction the inequality is false: for the face state
$\rho_g$ with Bloch vector $(1,1,1)/\sqrt3$,
$C(\rho_g\otimes\rho_g) \doteq \sqrt3/4 < \sqrt3-1 = 2C(\rho_g)$.
Corollaries~\ref{conj:equatorial} and~\ref{conj:selftensor} were stated as
conjectures in an earlier version of this work; the first, and the
inequality part of the second, were first proved by Zhang and
Haobo~\cite{ZhangHaobo2026}.

\begin{remark}[The deficit does not factorize]
\label{rem:deficit_counterexample}
An earlier version conjectured that, for equatorial $\rho$, the deficit
$D(\rho,\sigma)=(1+C(\rho))(1+C(\sigma))-1-C(\rho\otimes\sigma)$ equals
$C(\rho)f(\sigma)$ for a universal function $f$. This is false. For
$\vec r_\sigma=(0.7,0.4,0.2)$ and equatorial $\rho$ on the ray
$\vec r\propto(1,1,0)$, $D/C(\rho)=0.575$ for $C(\rho)\lesssim0.17$ but
$D/C(\rho)=0.5121$ at the $H$-type state $\vec r=(1,1,0)/\sqrt2$
(\Cref{fig:dichotomy}b; exact linear programming, two independent solvers).
For $30$ of $40$ random $\sigma$ with $s(\sigma)>0$, the ratio already
differs between three fixed equatorial test states $\rho$.
What survives is the following structure. Along any ray of equatorial
$\rho$ (Bloch vector $\lambda\vec n$, $\lambda$ increasing), $C(\rho)$ is
affine in $\lambda$ once positive, while $C(\rho\otimes\sigma)$ is the
$\ell_1$ distance from an affinely moving point to a polytope and hence
convex and piecewise linear in $\lambda$. Moreover
$C(\rho\otimes\sigma)\geq C(\sigma)$ for every $\rho$, because discarding
the first qubit is a stochastic map on phase space that preserves
$\Wfree$ (\Cref{thm:properties}(v)); with $D\geq0$ from
\Cref{prop:prod_upper}, this gives $D=0$ where $C(\rho)=0$. Hence $D$
is a concave, piecewise-linear function of $C(\rho)$ vanishing at
$C(\rho)=0$: $D/C(\rho)$ is constant on the first linear piece and
non-increasing thereafter, as in \Cref{fig:dichotomy}b. For some $\sigma$
(e.g.\ $\vec r_\sigma=(0.5,0.5,0.5)$, $(0.3,0.6,0.5)$, $(0.55,0.55,0.55)$,
$(0.85,0.3,0.2)$, and $(0.7,0.4,0.2)$ itself) the first-piece constant equals $(1+C(\sigma))/2$, but not
for all $\sigma$ we sampled. The witness bound
\eqref{eq:product_lower} gives
$D\leq(1+C(\rho))(1+C(\sigma))-\|W_\rho\|_1\|W_\sigma\|_1$ for all pairs.
\end{remark}

\begin{remark}[Origin of the asymmetry]
\label{rem:dichotomy_origin}
The dichotomy arises from the tetrahedral structure of the Wootters
product construction.
The phase-point operator $A_{(0,0)} = \tfrac{1}{2}(I+X+Y+Z)$ is diagonal
in the eigenbasis of $(X+Y+Z)/\sqrt{3}$, with eigenvalues
$(1\pm\sqrt3)/2$ and its positive eigenvector along the tetrahedral
direction $(1,1,1)$; the four phase-point
directions single out the sign-region $s>0$, while the dual tetrahedron
($s<0$) contains no phase-point direction.
The asymmetry is not intrinsic to the physics: the phase gate $S$ acts
on Bloch space as $(x,y,z)\mapsto(-y,x,z)$, reversing $s$, and
correspondingly $S\otimes I$ maps one branch of the dichotomy onto the
other --- the same frame anisotropy responsible for the
$\kappa^{Ry}=1$ versus $\kappa^{Rx}=2$ distinction
(\Cref{rem:cliff_equiv}).
\end{remark}

\section{The Tightness Ratio: Exact Results}
\label{sec:kappa}

\subsection{The Repetition-Code Subspace}

We work in $\mathcal{C} = \mathrm{span}\{\ket{00},\ket{11}\}$, the codespace
of the two-qubit repetition code with stabilizer $ZZ$, and use the logical
Pauli operators $\XL = X\otimes X$, $\YL = Y\otimes X$, $\ZL = Z\otimes I$.
This code has distance one ($\ZL$ has weight one) and corrects no errors;
we use it only as a convenient labelling of a two-dimensional subspace.
The six logical stabilizer states $\ket{0_L}=\ket{00}$, $\ket{1_L}=\ket{11}$,
$\ket{\pm_L}=(\ket{00}\pm\ket{11})/\sqrt2$ and
$\ket{\pm i_L}=(\ket{00}\pm i\ket{11})/\sqrt2$ are two-qubit stabilizer
states.
Preparation circuits for the three families studied below are shown in
\Cref{fig:circuits}.
Throughout, $c=\cos\theta$ and $s=\sin\theta$, and every Wigner function is
displayed as a $4\times4$ array with rows indexed by the first qubit's
phase point and columns by the second's, both in the order
$(q,p)=(0,0),(0,1),(1,0),(1,1)$.

\begin{figure}[ht]
\centering
\includegraphics[width=\columnwidth]{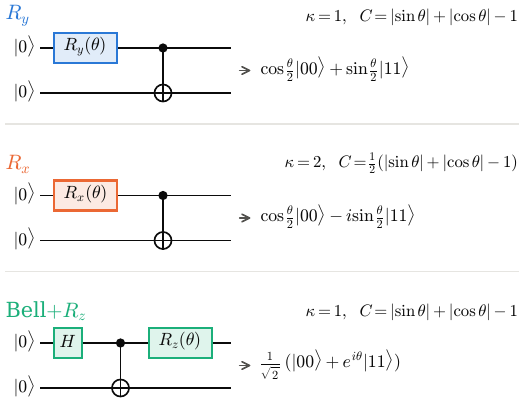}
\caption{%
  Preparation circuits for the three families.
  The CNOT encodes the single-qubit state into the codespace
  $\mathcal{C} = \mathrm{span}\{\ket{00},\ket{11}\}$.
  Real coherence ($R_y$, XZ meridian, blue) gives $\kappa=1$;
  imaginary coherence ($R_x$, YZ meridian, orange) gives $\kappa=2$;
  equatorial phase (Bell$+R_z$, equatorial circle, aqua) gives $\kappa=1$.%
}
\label{fig:circuits}
\end{figure}

\paragraph{Certificates.}
Each value of $C$ below is proved by a matching pair: an explicit free
state $\tau^*\in\mathrm{conv}(\Stab_2)$ with
$\|W_\rho-W_{\tau^*}\|_1$ equal to the claimed value (upper bound), and an
explicit Hermitian $H$ with $|\tr(HA_\alpha)|\leq4$ for all $\alpha$
(\Cref{thm:dual}) whose certified gap
$\tfrac14[\tr(H\rho)-\max_\sigma\tr(H\sigma)]$ equals the same value (lower
bound).
Two elementary facts are used repeatedly.
(a)~In the Wootters frame, $\tr(PA_\alpha)=\pm1$ for every Pauli $P$; for
qubit $k$ write $(x_k,y_k,z_k)\in\{\pm1\}^3$ for the signs of
$\tr(XA_{\alpha_k}),\tr(YA_{\alpha_k}),\tr(ZA_{\alpha_k})$, which satisfy
$x_ky_kz_k=+1$, so $y_k=x_kz_k$.
(b)~For a pure stabilizer state $\sigma$, $\tr(P\sigma)\in\{0,\pm1\}$ and
the Paulis with $\tr(P\sigma)\neq0$ pairwise commute; hence if the terms of
$H=\sum_P h_PP$ split into two sets with every term of one anticommuting with
every term of the other, $\tr(H\sigma)$ is bounded by the larger of
$\sum|h_P|$ over the two sets.
Similarly, the robustness has the dual bound
$\Gamma(\rho)\geq\tr(G\rho)$ for any $G$ with $|\tr(G\sigma)|\leq1$ for all
$\sigma\in\Stab_n$, since $\tr(G\rho)=\sum_i\alpha_i\tr(G\sigma_i)\leq\sum_i|\alpha_i|$.

\subsection{The $R_y$ and $R_x$ Families}

\begin{definition}
\begin{align}
  \ket{\psi_\theta^{Ry}} &:= \cos\tfrac{\theta}{2}\ket{00}
                              + \sin\tfrac{\theta}{2}\ket{11}, \\
  \ket{\psi_\theta^{Rx}} &:= \cos\tfrac{\theta}{2}\ket{00}
                              - i\sin\tfrac{\theta}{2}\ket{11},
\end{align}
for $\theta \in (0,\pi)$, tracing the XZ and YZ meridians of the
logical Bloch sphere respectively.
\end{definition}

\begin{theorem}[Exact $\kappa$ for $R_y$ and $R_x$]
\label{thm:kappa_ry_rx}
For all $\theta \in (0,\pi)$ (with $\kappa$ defined for
$\theta\neq\pi/2$, where both states are stabilizer states):
\begin{align}
  C(\rho_\theta^{Ry}) &= \sct - 1, &
  \kappa^{Ry} &= 1, \\[2pt]
  C(\rho_\theta^{Rx}) &= \tfrac{1}{2}\!\left(\sct - 1\right), &
  \kappa^{Rx} &= 2.
\end{align}
Both families share $\Gamma(\rho_\theta) = \sct$.
\end{theorem}

\begin{proof}
\emph{Reduction.} The logical Paulis $\XL$ (for $R_y$) and $\ZL\XL$ (for
$R_x$) map $\theta\mapsto\pi-\theta$ up to a global phase; they are
physical Pauli operators, hence phase-space translations that leave $C$ and
$\Gamma$ invariant. It therefore suffices to take $\theta\in(0,\pi/2]$, so
$c\geq0$, $s>0$ and $c+s\geq1$.
The Wigner functions are
\begin{equation*}
8W_{\rho^{Ry}} =
\begin{pmatrix}
1{+}c & 1{+}c & s & -s\\
1{+}c & 1{+}c & -s & s\\
s & -s & 1{-}c & 1{-}c\\
-s & s & 1{-}c & 1{-}c
\end{pmatrix}\!,
\end{equation*}
\begin{equation*}
8W_{\rho^{Rx}} =
\begin{pmatrix}
1{+}c{-}s & 1{+}c{+}s & 0 & 0\\
1{+}c{+}s & 1{+}c{-}s & 0 & 0\\
0 & 0 & 1{-}c{+}s & 1{-}c{-}s\\
0 & 0 & 1{-}c{-}s & 1{-}c{+}s
\end{pmatrix}\!.
\end{equation*}

\emph{$R_y$, upper bound.} Take
$\tau^*=(1-s)\ketbra{0_L}{0_L}+s\ketbra{+_L}{+_L}$. Direct computation gives
\begin{equation*}
8\bigl(W_{\rho^{Ry}}-W_{\tau^*}\bigr)=(c+s-1)
\begin{pmatrix}
J & 0\\ 0 & -J
\end{pmatrix}\!,
\end{equation*}
with $J$ the $2\times2$ all-ones block, so $\|W_{\rho^{Ry}}-W_{\tau^*}\|_1 = 8(c+s-1)/8 = c+s-1$.

\emph{$R_y$, lower bound.} Let
$H^{Ry}=2(I+ZZ)(\ZL+\XL)=2(ZI+IZ+XX-YY)$.
By fact (a), $\tr(H^{Ry}A_\alpha)=2[z_1+z_2+x_1x_2(1-z_1z_2)]\in\{\pm4\}$,
so the phase-space witness $S=\tfrac14\tr(H^{Ry}A_\alpha)$ takes values
$\pm1$. The terms split as $\{ZI,IZ\}$ and $\{XX,YY\}$, every member of one
set anticommuting with every member of the other, so by fact (b)
$\tr(H^{Ry}\sigma)\leq4$ for every stabilizer state. On the target,
$\langle ZI\rangle=\langle IZ\rangle=c$, $\langle XX\rangle=s$ and
$\langle YY\rangle=-s$, so $\tfrac14\tr(H^{Ry}\rho)=c+s$, and \Cref{thm:dual}
gives $C\geq c+s-1$.

\emph{$R_x$, upper bound.} Take
$\tau^*=(1-w)\ketbra{0_L}{0_L}+w\ketbra{-i_L}{-i_L}$ with
$w=(1+s-c)/2\in[0,1]$. Then
\begin{equation*}
8\bigl(W_{\rho^{Rx}}-W_{\tau^*}\bigr)=(c+s-1)
\begin{pmatrix}
F & 0\\ 0 & -F
\end{pmatrix}\!,
\end{equation*}
with $F=\left(\begin{smallmatrix}0&1\\1&0\end{smallmatrix}\right)$, which has four non-zero entries, so
$\|W_{\rho^{Rx}}-W_{\tau^*}\|_1=\tfrac12(c+s-1)$.

\emph{$R_x$, lower bound.} Let
$H^{Rx}=(I+ZZ)(\ZL-\YL)=ZI+IZ-YX-XY$. Here
$\tr(H^{Rx}A_\alpha)=(z_1+z_2)(1-x_1x_2)\in\{0,\pm4\}$, so $S\in\{0,\pm1\}$.
The terms split as $\{ZI,IZ\}$ and $\{YX,XY\}$, mutually anticommuting, so
$\tr(H^{Rx}\sigma)\leq2$. On the target $\langle YX\rangle=\langle XY\rangle=-s$,
so $\tfrac14\tr(H^{Rx}\rho)=(c+s)/2$ and $C\geq\tfrac12(c+s-1)$.

\emph{Robustness.} For $R_y$,
\begin{equation*}
\rho_\theta^{Ry}=c\,\ketbra{0_L}{0_L}+\tfrac{1-c+s}{2}\ketbra{+_L}{+_L}
-\tfrac{c+s-1}{2}\ketbra{-_L}{-_L}
\end{equation*}
has $\ell_1$ weight $c+s$, and $G=\ZL+\XL=ZI+XX$ is a sum of two
anticommuting Paulis, so $|\tr(G\sigma)|\leq1$ on stabilizer states and
$\Gamma\geq\tr(G\rho)=c+s$. Conjugating by the Clifford $S^\dagger\otimes I$
maps $\rho^{Ry}_\theta$ to $\rho^{Rx}_\theta$
(\Cref{rem:cliff_equiv}) and stabilizer states to stabilizer states, so
$\Gamma^{Rx}=\Gamma^{Ry}=c+s$.
Hence $\kappa^{Ry}=1$ and $\kappa^{Rx}=2$.
\end{proof}

\subsection{Where the Factor of 2 Comes From}

The certificates locate the factor of two exactly. For both families the
optimal difference $W_\rho-W_{\tau^*}$ has entries of the same magnitude,
$(c+s-1)/8$, but it is supported on eight phase points for $R_y$ and on four
for $R_x$ (\Cref{fig:wigner}). The Wigner functions themselves have four
negative entries for $R_y$ and two for $R_x$; this count is suggestive but
does not by itself determine $C$, because the free polytope $\Wfree$
contains Wigner functions with negative entries.
Both families have the same $\Gamma$ because $\Gamma$ is invariant under
all Clifford unitaries and the two families are Clifford-equivalent
(\Cref{rem:cliff_equiv}); $C$, a distance in a fixed frame, is not.

\begin{remark}[The two families are Clifford-equivalent]
\label{rem:cliff_equiv}
For all $\theta$,
\begin{equation}
  \rho_\theta^{Rx} = (S^\dagger\!\otimes I)\,\rho_\theta^{Ry}\,(S\otimes I),
\end{equation}
since $S^\dagger\ket{1} = -i\ket{1}$ maps the real coherence of
$\ket{\psi_\theta^{Ry}}$ onto the imaginary coherence of
$\ket{\psi_\theta^{Rx}}$.
The two families are therefore related by a depth-one, single-qubit,
free Clifford gate, and the factor of two between $C^{Rx}$ and $C^{Ry}$
is an instance of the non-invariance of $C$ under local Cliffords acting
on multiqubit systems (\Cref{thm:properties}(iii),
\Cref{rem:not_monotone}).
Any Clifford-invariant magic quantifier must coincide on the two
families---which is precisely why $\Gamma^{Rx}=\Gamma^{Ry}$.
The $\kappa$ trichotomy thus characterizes how the fixed Wigner frame
resolves different logical coherence directions, rather than a
difference between Clifford-inequivalent resources.
\end{remark}

\subsection{The Bell$+R_z$ Family: $\kappa = 1$}

\begin{definition}
\begin{equation}
  \ket{\Phi_\theta} := \frac{1}{\sqrt{2}}\bigl(\ket{00} + e^{i\theta}\ket{11}\bigr),
  \quad \theta \in [0,2\pi),
  \label{eq:BRz_family}
\end{equation}
tracing the equatorial circle of the logical Bloch sphere.
\end{definition}

The closed-form Wigner function is
\begin{equation}
  8\,W_{\Phi_\theta} =
  \begin{pmatrix}
    1{+}s & 1{-}s &  c & -c \\
    1{-}s & 1{+}s & -c &  c \\
     c    & -c    & 1{-}s & 1{+}s \\
    -c    &  c    & 1{+}s & 1{-}s
  \end{pmatrix},
  \label{eq:W_BRz}
\end{equation}
which has exactly 4 negative entries for every
$\theta \notin \{\pi/2, 3\pi/2\}$ (including the Bell states at
$\theta=0,\pi$) and none at $\theta=\pi/2,3\pi/2$.

\begin{theorem}[$\kappa$ for Bell$+R_z$]
\label{thm:kappa_BRz}
For all $\theta$ with $C(\rho_\theta) > 0$:
\begin{align}
  C(\rho_\theta)       &= |\sin\theta| + |\cos\theta| - 1, \\
  \Gamma(\rho_\theta) &= |\sin\theta| + |\cos\theta|, \\
  \kappa(\rho_\theta)  &= 1.
\end{align}
An optimal dual witness, piecewise constant in the quadrant of $\theta$, is
\begin{equation}
  H^{BRz}(\theta) = 2(I+ZZ)\bigl[\operatorname{sign}(\cos\theta)\,\XL
              + \operatorname{sign}(\sin\theta)\,\YL\bigr].
  \label{eq:BRz_witness}
\end{equation}
\end{theorem}

\begin{proof}
The Paulis $Z\otimes I$ ($\theta\mapsto\theta+\pi$) and $X\otimes X$
($\theta\mapsto-\theta$, up to a global phase) are phase-space translations,
so it suffices to take $\theta\in(0,\pi/2]$.
\emph{Upper bound:} $\tau^*=c\,\ketbra{+_L}{+_L}+(1-c)\ketbra{+i_L}{+i_L}$
gives
\begin{equation*}
8\bigl(W_{\Phi_\theta}-W_{\tau^*}\bigr)=(c+s-1)
\begin{pmatrix}
K & 0\\ 0 & -K
\end{pmatrix}\!,
\end{equation*}
with $K=\left(\begin{smallmatrix}1&-1\\-1&1\end{smallmatrix}\right)$, so $\|W_{\Phi_\theta}-W_{\tau^*}\|_1=c+s-1$.
\emph{Lower bound:} $H^{BRz}=2(XX-YY+YX+XY)$ has
$\tr(H^{BRz}A_\alpha)=2x_1x_2[1+z_1+z_2-z_1z_2]\in\{\pm4\}$, so $S=\pm1$;
its terms split as $\{XX,YY\}$ and $\{YX,XY\}$, mutually anticommuting, so
$\tr(H^{BRz}\sigma)\leq4$; and $\langle XX\rangle=-\langle YY\rangle=c$,
$\langle YX\rangle=\langle XY\rangle=s$ give $\tfrac14\tr(H^{BRz}\rho)=c+s$.
\emph{Robustness:}
$\rho_\theta=\tfrac{1+c-s}{2}\ketbra{+_L}{+_L}+s\,\ketbra{+i_L}{+i_L}
-\tfrac{c+s-1}{2}\ketbra{-_L}{-_L}$ has weight $c+s$, and
$G=\XL+\YL=XX+YX$ (anticommuting) gives $\Gamma\geq c+s$.
\end{proof}

\subsection{The Logical Bloch Sphere Trichotomy}

The three families trace distinct orbits on the logical Bloch sphere
(\Cref{fig:bloch}):
the $\ZL$ poles are stabilizer states ($C=0$);
the XZ meridian ($R_y$) gives $\kappa=1$;
the YZ meridian ($R_x$) gives $\kappa=2$;
and the equatorial circle (Bell$+R_z$) gives $\kappa=1$.
The $\kappa$ value is determined by the coherence \emph{direction}, not the
state's position along its orbit.

\begin{figure}[ht]
  \centering
  \includegraphics[width=0.82\columnwidth]{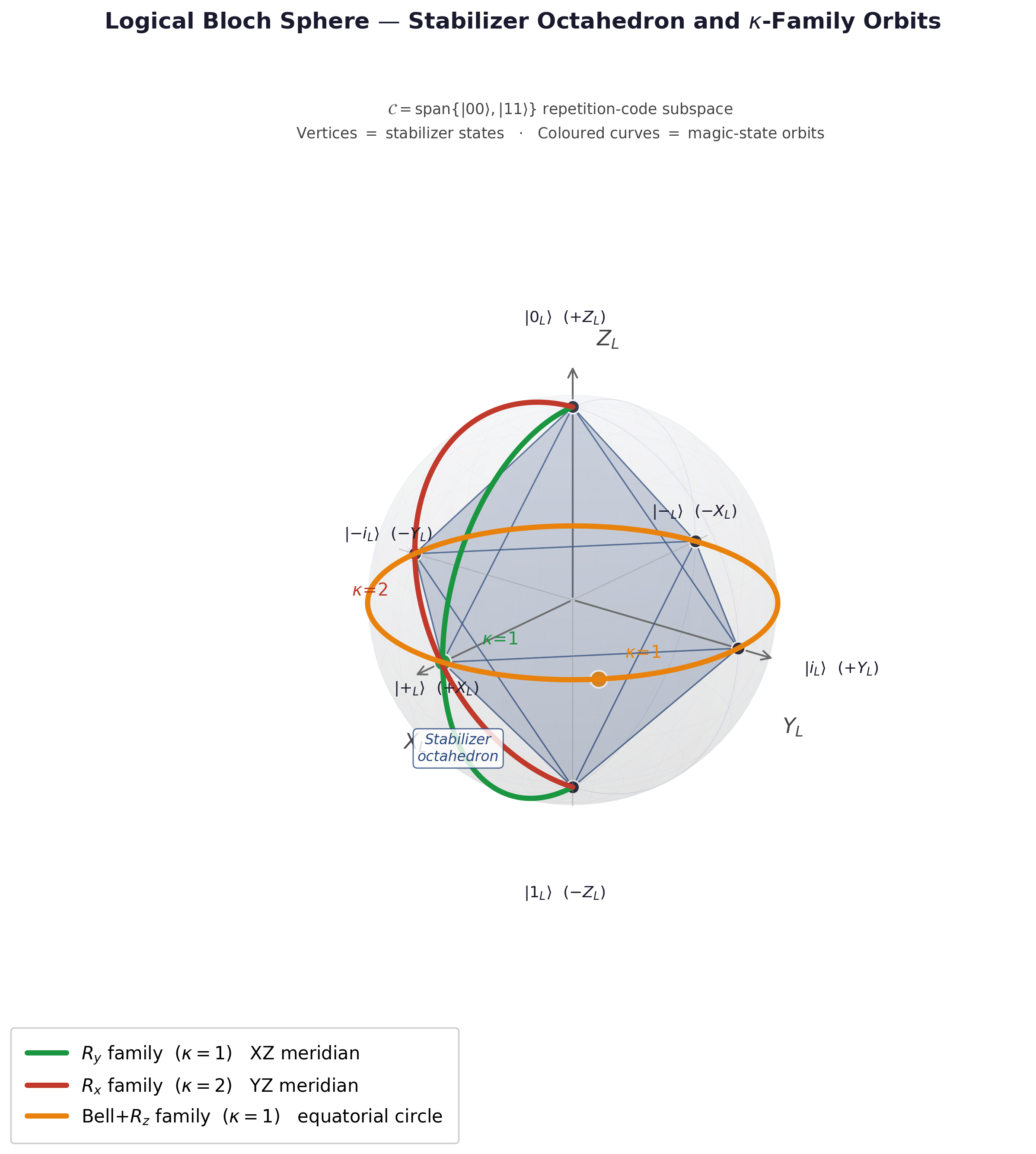}
  \caption{The logical Bloch sphere for
  $\mathcal{C} = \mathrm{span}\{|00\rangle,|11\rangle\}$.
  The inscribed octahedron has the six stabilizer states as vertices.
  $R_y$ (blue, $\kappa=1$), $R_x$ (orange, $\kappa=2$), and Bell$+R_z$
  (aqua, $\kappa=1$) are shown.}
  \label{fig:bloch}
\end{figure}

\subsection{Analytical Results Summary}

\Cref{tab:results} collects the exact results for the three families,
and \Cref{fig:kappa} shows $C$, $\Gamma-1$ and $\kappa$ along the $R_y$
and $R_x$ orbits.

\begin{table*}[ht]
\centering
\caption{Exact results for the three families, with
$c = \cos\theta$, $s = \sin\theta$; certificates are shown for
$\theta\in(0,\tfrac{\pi}{2}]$ (other quadrants follow by Pauli
translations, which replace $\ZL,\XL,\YL$ by $\pm$ themselves).
Kets in the $\tau^*$ row stand for their projectors.
``Support'' is the number of phase points on which the optimal difference
$W_\rho-W_{\tau^*}$ is non-zero.}
\label{tab:results}
\small
\setlength{\tabcolsep}{6pt}
\begin{tabular}{@{}llll@{}}
\toprule
 & $R_y$ & $R_x$ & Bell$+R_z$ \\
\midrule
$C(\rho_\theta)$
  & $|s|+|c|-1$
  & $\frac{1}{2}(|s|+|c|-1)$
  & $|s|+|c|-1$ \\[3pt]
$\Gamma(\rho_\theta)$
  & $|s|+|c|$
  & $|s|+|c|$
  & $|s|+|c|$ \\[3pt]
$\kappa(\rho_\theta)$
  & $1$
  & $2$
  & $1$ \\[3pt]
$\tau^*$
  & $(1{-}s)\ket{0_L}{+}s\ket{+_L}$
  & $\tfrac{1+c-s}{2}\ket{0_L}{+}\tfrac{1+s-c}{2}\ket{-i_L}$
  & $c\ket{+_L}{+}(1{-}c)\ket{+i_L}$ \\[3pt]
$H$
  & $2(I{+}ZZ)(\ZL{+}\XL)$
  & $(I{+}ZZ)(\ZL{-}\YL)$
  & $2(I{+}ZZ)(\XL{+}\YL)$ \\[3pt]
Support & 8 & 4 & 8 \\
Neg.\ entries of $W_\rho$  & 4 & 2 & 4 \\
Bloch orbit    & XZ meridian & YZ meridian & Equatorial \\
\bottomrule
\end{tabular}
\end{table*}

\begin{figure}[ht]
  \centering
  \includegraphics[width=0.48\columnwidth]{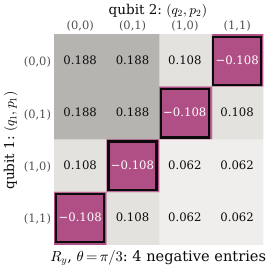}
  \hfill
  \includegraphics[width=0.48\columnwidth]{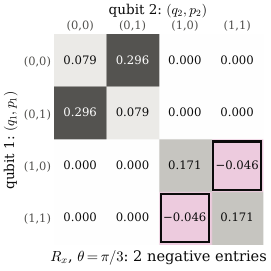}
  \caption{Wigner functions of $R_y$ (left) and $R_x$ (right)
  at $\theta = \pi/3$.
  The $R_y$ state has 4 negative entries; $R_x$ has only 2.
  The exact origin of $\kappa^{Ry} = 1$ vs.\ $\kappa^{Rx} = 2$ is the
  support of the optimal difference $W_\rho-W_{\tau^*}$ (8 vs.\ 4 phase
  points; \Cref{tab:results}).}
  \label{fig:wigner}
\end{figure}

\begin{figure*}[t]
  \centering
  \includegraphics[width=0.85\textwidth]{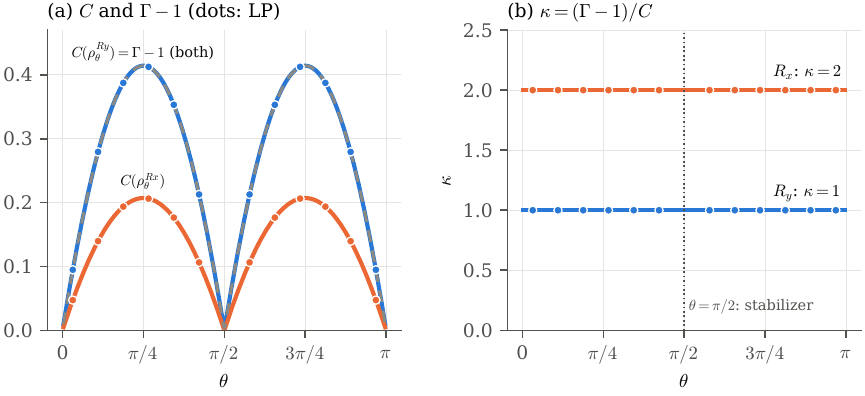}
  \caption{(a)~$C$ and $\Gamma-1$ along the $R_y$ (blue) and $R_x$ (orange)
  orbits; $\Gamma-1$ (grey dashed) is the same for both families and coincides
  with $C^{Ry}$. (b)~$\kappa=(\Gamma-1)/C$, equal to $1$ and $2$ for all
  $\theta\neq\pi/2$ (\Cref{thm:kappa_ry_rx}); at $\theta=\pi/2$ both
  states are stabilizer states. Dots are exact linear programs.}
  \label{fig:kappa}
\end{figure*}

\subsection{Connection to Noise Robustness}
\label{subsec:noise}

Under depolarizing noise $\Lambda_p(\rho) = (1-p)\rho + p\,\tfrac{I}{4}$ at
rate $p$, the $R_x$ family exhibits \emph{$\kappa$-rigidity} to linear-programming
precision (\Cref{fig:noise}; checked at seven noise levels for
$\theta=\pi/3$):
\begin{equation}
  \kappa\!\left(\Lambda_p(\rho_\theta^{Rx})\right) = 2
\end{equation}
for all $p$ below the critical noise level
\begin{equation}
  p_c = \frac{|\sin\theta|+|\cos\theta|-1}{|\sin\theta|+|\cos\theta|-\tfrac12},
\end{equation}
at which $C\to0$; numerically $\Gamma(\Lambda_p(\rho_\theta))-1
= (|\sin\theta|+|\cos\theta|-1) - p\,(|\sin\theta|+|\cos\theta|-\tfrac12)$
for both families, which vanishes exactly at $p_c$.
The factor-of-2 gap thus persists under depolarization for the $R_x$
family; we have no analytic proof of this rigidity.
The $R_y$ family, by contrast, is \emph{not} rigid: $\kappa(\Lambda_p(
\rho_\theta^{Ry}))$ rises from $1$ at $p=0$, reaches $\tfrac32$ at an
intermediate noise level ($p\approx0.155$ for $\theta=\pi/3$), and stays at
$\tfrac32$ up to $p_c$ (computed by linear programming over the full
two-qubit stabilizer polytope).
This asymmetry is itself frame structure invisible to $\Gamma$: the two
families share $\Gamma(\Lambda_p(\rho_\theta)) $ identically at every $p$
(both being Clifford images of the same state under $S^\dagger\otimes I$), yet their
$\kappa$ trajectories under the same noise diverge --- $C$ resolves a
difference in how depolarization acts on real versus imaginary logical
coherence that the Clifford-invariant $\Gamma$ cannot.

\begin{figure*}[t]
  \centering
  \includegraphics[width=0.85\textwidth]{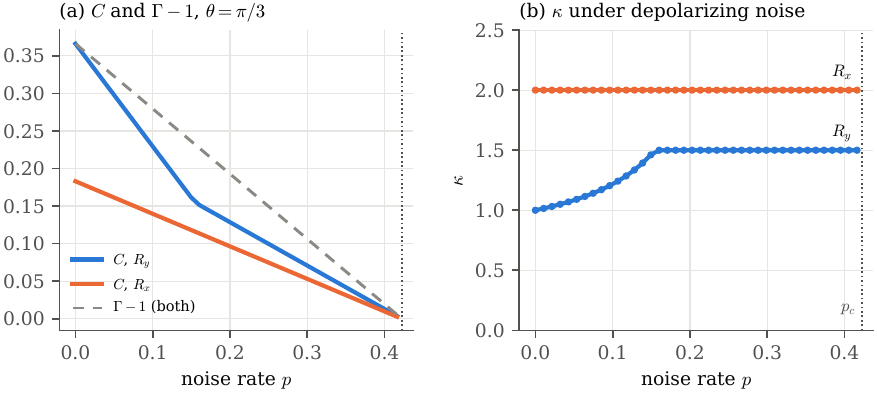}
  \caption{Depolarizing noise $\Lambda_p$ at $\theta=\pi/3$, by exact linear
  programming. (a)~$C$ for $R_y$ (blue) and $R_x$ (orange), and $\Gamma-1$
  (grey dashed, the same for both). (b)~$\kappa$: it stays at $2$ for $R_x$
  up to $p_c\approx0.423$, where $C\to0$, while for $R_y$ it rises from $1$
  and reaches $\tfrac32$ at $p\approx0.155$, despite the two families sharing
  $\Gamma$ at every $p$.}
  \label{fig:noise}
\end{figure*}

\section{Asymptotic Analysis}
\label{sec:asymptotic}

\subsection{Growth under Tensor Powers}

Because $C$ can grow exponentially under tensoring
(\Cref{cor:exact_powers}), the naive regularization
$\limsup_k \frac1k C(\rho^{\otimes k})$ is infinite already for the face
states. The meaningful rate is logarithmic:
\begin{equation}
  c^\infty(\rho) := \limsup_{k\to\infty}\frac1k\log\bigl(1+C(\rho^{\otimes k})\bigr).
\end{equation}
For every single-qubit $\rho$, \Cref{prop:prod_upper} and
\Cref{thm:product_equality} give
\begin{equation*}
\log\|W_\rho\|_1\ \leq\ c^\infty(\rho)\ \leq\ \log(1+C(\rho)),
\end{equation*}
so $c^\infty(\rho)=\log(1+C(\rho))$ exactly whenever $s(\rho)\leq0$
(\Cref{conj:selftensor}); for the face state $\rho_f$ this is
$\log\frac{1+\sqrt3}{2}$. For $s(\rho)>0$ and $C(\rho)>0$ the two bounds
differ. For the
opposite face state $\rho_g$ ($s=+1$, $\|W_{\rho_g}\|_1=1$) the computed
values $C(\rho_g^{\otimes2})=\sqrt3/4$ and $C(\rho_g^{\otimes3})\approx0.5245$
lie far below $(1+c_0)^k-1$; whether $c^\infty(\rho_g)<\log(1+c_0)$ is open.

\subsection{Why $C$-Based Distillation Bounds Fail}

\begin{remark}[$C$ is not a magic monotone]
\label{rem:not_monotone}
$H \otimes I$ (a free Clifford) increases $C$ for roughly $40\%$ of
Haar-random two-qubit pure states (770 of 2000 in one sample), by more
than $0.3$ in the worst cases.
More sharply, the local gate $S^\dagger\otimes I$ maps the $R_y$ family onto the
$R_x$ family while halving $C$ (\Cref{rem:cliff_equiv}).
Consequently, $C$ cannot support distillation-rate bounds, which require
monotonicity under stabilizer operations; $\Gamma$~\cite{Howard2017} and
the monotones of Refs.~\cite{Seddon2021,WangWildeSu2020,LeoneBittel2024}
can.
\end{remark}

Nor is $\kappa$ controlled under tensor powers in general: $\Gamma$ is
submultiplicative (tensor two optimal decompositions), but for $\rho_f$ the resulting bound
$\kappa(\rho_f^{\otimes k})\leq(3^{k/2}-1)/(((1+\sqrt3)/2)^k-1)$ diverges and
is uninformative. Along the GHZ direction of \Cref{subsec:ghz_n}, by
contrast, $\kappa$ provably decays exponentially in the number of qubits.

\section{Discussion}
\label{sec:discussion}

\subsection{Summary of Contributions}
\begin{enumerate}[(i)]
  \item The bound $\Gamma \geq 1 + C/M_n$ and the tightness ratio
        $\kappa = (\Gamma-1)/C$ as a measure of how far the geometric
        picture departs from the operational one.
  \item Exact values $\kappa = 1, 2, 1$ for the $R_y$, $R_x$,
        and Bell$+R_z$ families, each proved by an explicit optimal free
        state and an explicit normalized dual witness; the factor of two is
        the ratio of the supports of the optimal differences (8 vs.\ 4
        phase points).
  \item The single-qubit closed form $C=\tfrac12\max(0,\|\vec r\|_1-1)$
        (\Cref{prop:closed_form}) and the tetrahedral dichotomy: the
        product formula is an upper bound for all single-qubit factors
        (\Cref{prop:prod_upper}), attained whenever every factor has
        non-positive Bloch sign product (\Cref{thm:product_equality},
        extending Zhang and Haobo~\cite{ZhangHaobo2026}), and numerically
        not attained otherwise (\Cref{obs:dichotomy}); the deficit does not
        factorize (\Cref{rem:deficit_counterexample}).
  \item Code-projected witnesses: the optimal witnesses are logical Pauli
        combinations compressed to the codespace; for $R_y$ and
        Bell$+R_z$ the bare logical operators certify only half of $C$
        (\Cref{thm:logical_witnesses}).
  \item GHZ families for all $n$ (\Cref{thm:ghz}): $\Gamma$ is exact for
        every $n$, and a phase-parity lemma for the Wootters phase-point
        operators yields $C\geq2^{\lfloor n/2\rfloor-1}(|\sin\phi|+|\cos\phi|-1)$,
        tight for $n=2$ and numerically for $n=3,4$, so
        $\kappa\leq2^{1-\lfloor n/2\rfloor}$.
  \item Exact maxima of $C$ (\Cref{sec:maxC}): the single-qubit
        maximum is $(\sqrt{3}-1)/2$, attained precisely at the eight
        face states, and $C(\rho_f\otimes\rho_f) = \sqrt{3}/2$ exactly;
        the tensor-power values are exact,
        $C(\rho_f^{\otimes n}) = ((1+\sqrt{3})/2)^n - 1$, and the
        maximum over fully separable states is determined
        (\Cref{cor:exact_powers}); these hold in every product frame.
  \item The frame freedom (\Cref{rem:frame_freedom}): the two
        single-qubit frames give identical $C$; the $2048$ two-qubit
        Pauli-covariant frames give $C(\rho^{Rx})/C(\rho^{Ry})\in
        \{\tfrac12,1,2\}$, with the inversion forced by a local phase gate;
        and $M_2=2$ in every frame by the Mermin--Peres obstruction, so
        $\kappa\geq\tfrac12$ is frame independent for two qubits.
  \item $C$ is not a monotone, so it cannot replace $\Gamma$ in
        distillation bounds.
\end{enumerate}

\subsection{Code-Projected Logical Witnesses}
\label{subsec:qec}

All three two-qubit families live in
$\mathcal{C} = \mathrm{span}\{\ket{00},\ket{11}\}$, with code projector
$\Pi=\tfrac12(I+ZZ)$.

\begin{theorem}[Code-projected logical witnesses]
\label{thm:logical_witnesses}
For $\theta$ in the first quadrant, optimal dual witnesses for the three
families are
\begin{align}
  H^{Ry} &= 2(I+ZZ)(\ZL+\XL) = 4\,\Pi(\ZL+\XL)\Pi, \\
  H^{Rx} &= (I+ZZ)(\ZL-\YL) = 2\,\Pi(\ZL-\YL)\Pi, \\
  H^{BRz} &= 2(I+ZZ)(\XL+\YL) = 4\,\Pi(\XL+\YL)\Pi,
\end{align}
with the signs of $\ZL$, $\XL$, $\YL$ adjusted to the quadrant as in
\Cref{thm:kappa_ry_rx,thm:kappa_BRz}.
Each commutes with $ZZ$ and is supported on the codespace, so on code
states $\tr(H\rho)$ depends only on the logical Bloch vector
$(\langle\XL\rangle,\langle\YL\rangle,\langle\ZL\rangle)$; for the
Bell$+R_z$ family
\begin{equation}
  C(\rho_\theta) = |\langle\XL\rangle| + |\langle\YL\rangle| - 1,
\end{equation}
the $\ell_1$ distance from the logical Bloch vector to the logical
stabilizer octahedron.
By contrast, the bare logical operators $L=\ZL+\XL$ and $\XL+\YL$ satisfy
$\max_\alpha|\tr(LA_\alpha)|=2$, so their admissible multiples $2L$
certify only half of $C^{Ry}$ and $C^{BRz}$; for $R_x$ the bare operator
$2(\ZL-\YL)$ is itself optimal.
\end{theorem}

\begin{proof}
The identities $(I+ZZ)L=2\Pi L\Pi$ hold because each logical Pauli commutes
with $ZZ$; optimality is \Cref{thm:kappa_ry_rx,thm:kappa_BRz}. For the bare
operators, fact (a) of \Cref{sec:kappa} gives
$\tr((ZI+XX)A_\alpha)=z_1+x_1x_2$, $\tr((ZI-YX)A_\alpha)=z_1(1-x_1x_2)$
and $\tr((XX+YX)A_\alpha)=x_1x_2(1+z_1)$, each of maximum modulus $2$, so
the admissible multiple is $2L$. Each $L$ is a sum of two anticommuting
Paulis, so $\max_\sigma\tr(L\sigma)=1$, and the certified gap
$\tfrac14[\tr(2L\rho)-2]$ equals $\tfrac12(c+s-1)$ in all three cases: half
of $C^{Ry}$ and $C^{BRz}$, and all of $C^{Rx}$.
\end{proof}

The stabilizer-multiplied terms ($IZ$ and $YY$ for $R_y$) do not change
the logical action of the witness, but they redistribute its weight over
phase space so that the normalization constraint of \Cref{thm:dual} is
met at full strength. In this sense the optimal witness is a logical
observable compressed to the codespace rather than a bare logical
operator.

\begin{remark}[Scope]
\label{rem:fault_tolerant_scope}
The two-qubit repetition code has distance one and corrects no errors, so
no statement about fault tolerance follows from
\Cref{thm:logical_witnesses}. What does persist beyond two qubits is the
code-commuting structure: for the $n=3,4$ analogues of the three families
(codewords $\ket{0}^{\otimes n}$, $\ket{1}^{\otimes n}$, stabilizers
$Z_iZ_{i+1}$), the dual linear program restricted to operators commuting
with every code stabilizer attains $C$ exactly at all sampled angles.
Optimal witnesses are not unique: already for $R_x$ at $n=2$ the linear
program also returns optimal witnesses containing $XI$, $YI$ and $XZ$,
which anticommute with $ZZ$.
\end{remark}

\begin{figure}[ht]
  \centering
  \includegraphics[width=0.82\columnwidth]{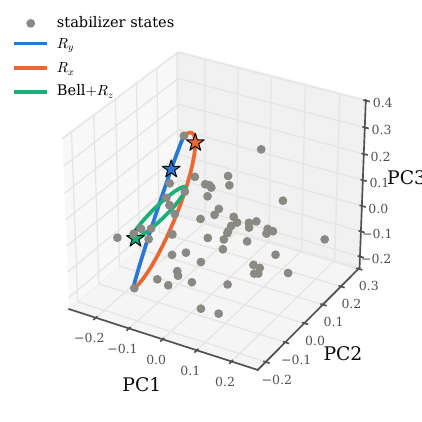}
  \caption{PCA of 16-dimensional Wigner vectors (first three components).
  The 60 stabilizer states (grey) span the free polytope; the three
  families trace distinct orbits, with stars at the most magic point of each
  ($\theta=\pi/4$). The projection is for orientation only: $C(\rho)$ is
  the $\ell_1$ distance in the full 16-dimensional Wigner space, which the
  projection does not preserve.}
  \label{fig:pca}
\end{figure}

\subsection{$n$-Qubit GHZ Families}
\label{subsec:ghz_n}

The three families of \Cref{sec:kappa} generalize to $n$ qubits.
The $n$-qubit GHZ$+R_z$ family is
\begin{equation}
  \ket{\Phi_\phi^{(n)}} := \frac{\ket{0}^{\otimes n}
    + e^{i\phi}\ket{1}^{\otimes n}}{\sqrt{2}},
  \quad \phi \in [0,2\pi),
  \label{eq:ghz_rz}
\end{equation}
prepared by $R_z(\phi)\otimes I^{\otimes(n-1)}$ after an $n$-qubit GHZ circuit.
On the span of $\ket{0}^{\otimes n}$ and $\ket{1}^{\otimes n}$ the operators
$X^{\otimes n}$, $Y\otimes X^{\otimes(n-1)}$ and $Z\otimes I^{\otimes(n-1)}$
act as logical $X$, $Y$, $Z$, and
$\langle X^{\otimes n}\rangle=\cos\phi$,
$\langle Y\otimes X^{\otimes(n-1)}\rangle=\sin\phi$.
The states $\ket{\mathrm{GHZ}_\psi}:=\ket{\Phi^{(n)}_\psi}$ with
$\psi\in\{0,\tfrac\pi2,\pi,\tfrac{3\pi}{2}\}$ are stabilizer states, with
stabilizer groups generated by $\pm X^{\otimes n}$ or
$\pm Y\otimes X^{\otimes(n-1)}$ together with the code stabilizers
$Z_iZ_{i+1}$.
Let $G_n=\{Z_E:E\subseteq[n],\ |E|\text{ even}\}$ denote the code
stabilizer group, $|G_n|=2^{n-1}$.

\begin{proposition}[Phase parity]
\label{prop:phase_parity}
The off-diagonal Wigner matrix element is
\begin{equation}
  \langle 0^n | A_\alpha | 1^n \rangle
  = (-1)^P \cdot 2^{-n/2} \cdot e^{i\pi(2m-n)/4},
  \label{eq:off_diag_mel}
\end{equation}
where $m = |\{k : q_k=1\}|$ and $P = \sum_k p_k$.
The phase $\vartheta_m := \pi(2m-n)/4$ is a multiple of $\tfrac{\pi}{2}$
for even $n$ and an odd multiple of $\tfrac{\pi}{4}$ for odd $n$.
\end{proposition}

\begin{proof}
The single-qubit matrix element is
$\langle 0|A_{(q,p)}|1\rangle = \tfrac{1}{2}(-1)^p(1-i(-1)^q)$.
Using $(1+i)=\sqrt{2}e^{i\pi/4}$ and $(1-i)=\sqrt{2}e^{-i\pi/4}$,
the $n$-qubit product gives
$\langle 0^n|A_\alpha|1^n\rangle
 = 2^{-n}(-1)^P(1-i)^{n-m}(1+i)^m
 = (-1)^P 2^{-n/2} e^{i\pi(2m-n)/4}$.
Since $2m-n$ has the same parity as $n$, the claim follows.
\end{proof}

The same phases control the normalization of the natural witness for this
family.

\begin{lemma}[Witness normalization]
\label{lem:ghz_norm}
For $a,b\in\{\pm1\}$ let
\begin{equation}
  H_n(a,b) := \sum_{g\in G_n} g\,\bigl(a\,X^{\otimes n}
             + b\,Y\otimes X^{\otimes(n-1)}\bigr).
  \label{eq:ghz_witness_n}
\end{equation}
With the sign notation of fact (a) in \Cref{sec:kappa},
\begin{equation}
  \tr\bigl(H_n(a,b)\,A_\alpha\bigr)
  = \Bigl(\prod_k x_k\Bigr)\,2^{n/2}\,\bigl(a\cos\vartheta'_m+b\sin\vartheta'_m\bigr),
\end{equation}
where $\vartheta'_m=\pi(n-2m)/4$ and $m=|\{k:q_k=1\}|$, and hence
$\mu_n:=\max_\alpha|\tr(H_n(a,b)A_\alpha)| = 2^{\lceil n/2\rceil}$.
\end{lemma}

\begin{proof}
For $E\subseteq[n]$, $Z_EX^{\otimes n}$ acts as $ZX=iY$ on the qubits in
$E$ and as $X$ elsewhere, so by fact (a)
$\tr(Z_EX^{\otimes n}A_\alpha)=\prod_kx_k\prod_{k\in E}(iz_k)$.
Summing over even $E$ selects the real part:
$\sum_{|E|\,\mathrm{even}}\prod_{k\in E}(iz_k)=\mathrm{Re}\prod_k(1+iz_k)$.
For the second term, $Y=iXZ$ gives
$Z_E\,Y\otimes X^{\otimes(n-1)}=-i\,Z_{E\triangle\{1\}}X^{\otimes n}$, and
$E\triangle\{1\}$ ranges over the odd subsets, so this sum is
$\prod_kx_k\,\mathrm{Im}\prod_k(1+iz_k)$.
Finally $z_k=(-1)^{q_k}$, so
$\prod_k(1+iz_k)=(1+i)^{n-m}(1-i)^m=2^{n/2}e^{i\vartheta'_m}$.
For even $n$, $\vartheta'_m$ is a multiple of $\tfrac\pi2$ and
$|a\cos\vartheta'_m+b\sin\vartheta'_m|=1$ at every phase point; for odd
$n$, $\vartheta'_m$ is an odd multiple of $\tfrac\pi4$ and the bracket takes
the values $0$ and $\pm\sqrt2$, the latter attained because $m$ ranges over
at least three consecutive integers.
\end{proof}

\begin{theorem}[GHZ$+R_z$ for all $n$]
\label{thm:ghz}
For every $n\geq2$ and every $\phi$,
\begin{align}
  \Gamma\bigl(\Phi_\phi^{(n)}\bigr) &= |\sin\phi|+|\cos\phi|,
  \label{eq:ghz_gamma}\\
  C\bigl(\Phi_\phi^{(n)}\bigr) &\geq 2^{\lfloor n/2\rfloor-1}
     \bigl(|\sin\phi|+|\cos\phi|-1\bigr).
  \label{eq:ghz_lower}
\end{align}
Equality holds in~\eqref{eq:ghz_lower} for $n=2$ (\Cref{thm:kappa_BRz}),
and for $n=3,4$ at every angle we evaluated by exact linear programming
over all $1080$ and $36{,}720$ stabilizer states.
Consequently $\kappa(\Phi_\phi^{(n)})\leq2^{1-\lfloor n/2\rfloor}$ for all
$n$, with equality for $n=2$ and at the evaluated angles for $n=3,4$.
\end{theorem}

\begin{proof}
The Paulis $Z\otimes I^{\otimes(n-1)}$ ($\phi\mapsto\phi+\pi$) and
$X^{\otimes n}$ ($\phi\mapsto-\phi$, up to a global phase) are phase-space
translations, so take $\phi\in(0,\tfrac\pi2]$ and write $c=\cos\phi$,
$s=\sin\phi$.

\emph{Robustness.} All states involved lie in the two-dimensional span of
$\ket{0}^{\otimes n},\ket{1}^{\otimes n}$, where the logical Bloch vectors
of $\mathrm{GHZ}_0$, $\mathrm{GHZ}_{\pi/2}$, $\mathrm{GHZ}_\pi$ are
$(1,0,0)$, $(0,1,0)$, $(-1,0,0)$; hence
\begin{equation*}
\rho=\tfrac{1+c-s}{2}\,\mathrm{GHZ}_0+s\,\mathrm{GHZ}_{\pi/2}
-\tfrac{c+s-1}{2}\,\mathrm{GHZ}_\pi
\end{equation*}
with $\ell_1$ weight $c+s$. The operator
$X^{\otimes n}+Y\otimes X^{\otimes(n-1)}$ is a sum of two anticommuting
Paulis, so its expectation on any stabilizer state has modulus at most $1$,
while its expectation on $\rho$ is $c+s$; hence $\Gamma=c+s$.

\emph{Wigner distance.} Take $H=(2^n/\mu_n)\,H_n(1,1)$, which satisfies
$|\tr(HA_\alpha)|\leq2^n$ by \Cref{lem:ghz_norm}, and apply
\Cref{thm:dual}. Every $g\in G_n$ fixes $\rho$, so
$\tr(H_n(1,1)\rho)=2^{n-1}(c+s)$. The $2^n$ Hermitian Pauli terms of
$H_n(1,1)$ split into $\{gX^{\otimes n}\}$ and
$\{g\,Y\otimes X^{\otimes(n-1)}\}$, and every member of the first set
anticommutes with every member of the second (the $g$'s commute with
both $X^{\otimes n}$ and $Y\otimes X^{\otimes(n-1)}$, which anticommute with
each other). By fact (b), $\tr(H_n(1,1)\sigma)\leq2^{n-1}$ for every
stabilizer state, with equality at $\mathrm{GHZ}_0$. Therefore
\begin{align*}
C(\rho)&\geq\frac{2^{n-1}(c+s)-2^{n-1}}{\mu_n}
=2^{\,n-1-\lceil n/2\rceil}(c+s-1)\\
&=2^{\lfloor n/2\rfloor-1}(c+s-1).
\end{align*}
The bound on $\kappa$ follows from~\eqref{eq:ghz_gamma}.
\end{proof}

\begin{remark}[What the parity does]
The factor $2^{\lfloor n/2\rfloor-1}$ equals $1$ for $n=2,3$ and $2$ for
$n=4,5$: the doubling first seen at $n=4$ is the first step of an
exponential sequence, not an even/odd alternation. Phase parity enters
through the witness normalization. For odd $n$ the logical-$X$ and
logical-$Y$ parts of $H_n$ add at some phase points (the
$\sqrt2$ in \Cref{lem:ghz_norm}), which costs a factor $\sqrt2$ relative
to even $n$; for even $n$ only one part survives at each phase point.
Combined with $M_n\geq2^{\lfloor n/2\rfloor}$, the result is consistent
with $\kappa\geq1/M_n$.
\end{remark}

\begin{conjecture}[Tightness for all $n$]
\label{conj:ghz_parity}
Equality holds in~\eqref{eq:ghz_lower} for every $n\geq2$, i.e.\
$C(\Phi_\phi^{(n)}) = 2^{\lfloor n/2\rfloor-1}(|\sin\phi|+|\cos\phi|-1)$ and
$\kappa(\Phi_\phi^{(n)})=2^{1-\lfloor n/2\rfloor}$.
\end{conjecture}

An earlier version of this work conjectured $C=|\sin\phi|+|\cos\phi|-1$ for
all odd $n$ and twice that for all even $n$. That conjecture is
inconsistent with the proved case $n=2$ and is refuted at $n=5$ by
\Cref{thm:ghz}.

For the $n$-qubit $R_y$ and $R_x$ analogues,
$\cos\tfrac\theta2\ket{0}^{\otimes n}+\sin\tfrac\theta2\ket{1}^{\otimes n}$
and
$\cos\tfrac\theta2\ket{0}^{\otimes n}-i\sin\tfrac\theta2\ket{1}^{\otimes n}$,
exact linear programming gives $C=|\sin\theta|+|\cos\theta|-1$ and
$\Gamma=|\sin\theta|+|\cos\theta|$ at $n=3$ and $4$ for every angle we
evaluated. Hence $\kappa=1$ for both, and the factor of two that separates
$R_x$ from $R_y$ at $n=2$ does not persist.

\subsection{Open Problems}
\label{subsec:open}

\begin{enumerate}[(i)]
\item \textbf{Repetition-code families.}
  The values of $\kappa$ obtained by exact linear programming over the
  complete $n$-qubit stabilizer polytope are:

  \begin{center}\small\setlength{\tabcolsep}{4pt}
  \begin{tabular}{@{}lccc@{}}
  \toprule
  Family & $n=2$ & $n=3$ & $n=4$ \\
  \midrule
  $R_y$      & $1$ & $1$ & $1$   \\
  $R_x$      & $2$ & $1$ & $1$   \\
  Bell$+R_z$ & $1$ & $1$ & $1/2$ \\
  \bottomrule
  \end{tabular}\\[2pt]
  \footnotesize Stabilizer-state counts $6, 60, 1080, 36720$ for
  $n=1,\dots,4$; the $n=5$ polytope has $2{,}423{,}520$ vertices, beyond
  our exact linear-programming capacity. The $n=2$ row is proved
  (\Cref{sec:kappa}); the Bell$+R_z$ row satisfies
  $\kappa\leq2^{1-\lfloor n/2\rfloor}$ for all $n$ (\Cref{thm:ghz}).
  \end{center}

  Is $\kappa=1$ for the $R_y$ and $R_x$ analogues at every $n\geq3$?
  Which codes admit optimal witnesses in the commutant of their
  stabilizer group, as found here for $n\leq4$?

\item \textbf{Classification of integer-$\kappa$ states.}
  Which states have integer $\kappa$, and is there a combinatorial
  condition on the support of the optimal difference $W_\rho-W_{\tau^*}$
  that decides it?

\item \textbf{The strict branch and the deficit.}
  Prove \Cref{obs:dichotomy}, and determine the deficit for
  $s(\sigma)>0$. It is not of the form $C(\rho)f(\sigma)$
  (\Cref{rem:deficit_counterexample}), but along rays of equatorial $\rho$
  it is concave and piecewise linear in $C(\rho)$; the slopes and
  breakpoints are unknown. Exact primal and dual
  certificates on a rational grid of pairs, which Zhang and Haobo have
  proposed to compute, should identify the pieces and their active
  stabilizer mixtures.

\item \textbf{Maximum of $C$ for two and more qubits.}
  The single-qubit maximum is settled:
  $\max_\rho C(\rho) = (\sqrt3-1)/2$, attained exactly at the eight
  face states (\Cref{thm:maxC1}), and
  $C(\rho_f\otimes\rho_f) = \sqrt3/2$ exactly (\Cref{prop:maxC2});
  this maximizer has $6$ negative Wigner entries.
  By \Cref{cor:exact_powers} the maximum over fully separable states
  equals $\sqrt3/2$ for two qubits and $((1+\sqrt3)/2)^n - 1$ in
  general; what remains open is whether any \emph{entangled} state
  exceeds it.
  Numerical optimization over two-qubit pure states finds none.
  In the trace-distance metric, Ref.~\cite{Palhares2026} likewise
  identifies the product of two face states as the leading numerical
  candidate for the two-qubit maximum, while for three qubits their
  sampling favours the \emph{entangled} Hoggar state. For $C$ the
  ordering at three qubits is reversed. Because $C$ is invariant under
  Pauli translations but not under all Cliffords, it is not constant on
  the Clifford orbit of Hoggar fiducials: sampling that orbit gives values
  between $0.75$ and $1.375$ (the most frequent being
  $103/84\approx1.2262$), all below $1.5490 = C(\rho_f^{\otimes3})$,
  so in the Wigner metric the product of faces dominates every sampled
  Hoggar state.
  For the stabilizer R\'{e}nyi entropy the two-qubit maximisers are known
  numerically to be \emph{entangled}---the Weyl--Heisenberg MUB fiducial
  states~\cite{Maximal2026}---so the answer to the maximisation problem is
  measure dependent, and our separable result does not by itself rule out
  entangled states exceeding it for $C$.
  We note the identity $(\sqrt3-1)/2 = -\lambda_{\min}(A_{+++})$,
  relating the single-qubit maximum to the most negative eigenvalue
  of a vertex of the dual $\Lambda_1$ polytope~\cite{Zurel2020}.

\item \textbf{GHZ tightness.}
  Prove \Cref{conj:ghz_parity}, i.e.\ exhibit for every $n$ a free state
  $\tau^*$ with $\|W_{\Phi^{(n)}_\phi}-W_{\tau^*}\|_1=2^{\lfloor n/2\rfloor-1}
  (|\sin\phi|+|\cos\phi|-1)$. If it holds, the Wigner distance overstates
  the robustness-based simulation cost of GHZ phase states by a factor
  growing as $2^{n/2}$.
\end{enumerate}

\subsection{Broader Context}

The tightness ratio $\kappa$ quantifies how far a frame-dependent
geometric distance departs from the operational robustness. Within the
two-qubit repetition codespace it takes exact integer values set by the
logical coherence direction; along GHZ phase states it provably decays
exponentially with $n$. A fixed-frame Wigner distance can therefore be an
exact and certifiable geometric quantity while being a poor proxy for
simulation cost, and the size of the gap is itself a frame-dependent,
computable quantity. The certificates constructed here---explicit free
states and normalized witnesses---make every value of $C$ quoted for these
families checkable by hand, which is what makes such comparisons
reliable.

\section*{Acknowledgments}
The authors thank
Chandan Datta (IISER Kolkata),
Amit Kundu (SNBNCBS),
Snehasish Roy Chowdhury (ISI Kolkata),
Subhendu B. Ghosh (SNBNCBS),
Ambuj (IIT Jodhpur), and
Jai Lalita (IIT Jodhpur)
for valuable discussions and helpful comments, and Wenlin Zhang and Haobo
(National University of Singapore) for sharing their proof and Lean
formalization of the equatorial case of \Cref{thm:product_equality}.
All scientific content, mathematical results, and conclusions are the
sole responsibility of the authors. This work received no external funding.
\paragraph{AI disclosure statement.}
During the preparation of this manuscript, the authors used Claude
(Anthropic) to support the checking of proofs, the writing of code for the
numerical verification and figures, and \LaTeX{} editing. The authors
reviewed, validated and critically assessed all AI-assisted outputs and take
full responsibility for the content of the manuscript.

\section*{Conflict of interest}
The authors declare no competing interests.

\section*{Data availability statement}
No experimental data were generated in this work. The code that reproduces
all numerical results (linear programs, frame census and figures) is
available at \url{https://github.com/Soumyo2121/A-Phase-Space-Geometric-Measure-of-Magic-in-Qubit-Systems}.

\bibliographystyle{unsrt}
\bibliography{preprint}

\appendix

\section{Structural Properties of $C$}
\label{app:properties}

\subsection{Injectivity of the Discrete Wigner Transform}

\begin{lemma}[Injectivity]
\label{lem:injectivity}
The discrete Wigner transform $\rho \mapsto W_\rho$ is linear and injective.
\end{lemma}
\begin{proof}
Linearity is immediate.
For injectivity, suppose $W_\rho(\alpha) = 0$ for all $\alpha$.
Since $\{A_\alpha\}$ form an orthogonal basis of the space of $2^n\times 2^n$
Hermitian operators under the Hilbert--Schmidt inner product, any operator
admits $\rho = 2^{-n}\sum_\alpha \tr(\rho A_\alpha) A_\alpha$.
If all coefficients vanish then $\rho = 0$.
\end{proof}

\subsection{Proofs of Theorem~\ref{thm:properties}}

\begin{proof}[Faithfulness (i)]
$(\Rightarrow)$ $C(\rho)=0$ implies $W_\rho \in \Wfree$ (the minimum is
attained on the compact set $\Wfree$), so $W_\rho=W_\tau$ for some
$\tau\in\mathrm{conv}(\Stab_n)$; by linearity and injectivity,
$\rho=\tau\in\mathrm{conv}(\Stab_n)$.
$(\Leftarrow)$ $\rho \in \mathrm{conv}(\Stab_n)$ gives $W_\rho\in\Wfree$,
so $C(\rho)=0$ by taking $W_f = W_\rho$.
\end{proof}

\begin{proof}[Convexity (ii)]
Let $W_{f_i}$ be optimal for $\rho_i$.
Since $\Wfree$ is convex, $W_f := pW_{f_1}+(1-p)W_{f_2} \in \Wfree$.
By linearity of $W$ and the triangle inequality:
\begin{align*}
  &C\bigl(p\rho_1+(1-p)\rho_2\bigr) \\
  &\;\leq \bigl\|p(W_{\rho_1}{-}W_{f_1}) + (1{-}p)(W_{\rho_2}{-}W_{f_2})\bigr\|_1 \\
  &\;\leq p\,C(\rho_1)+(1-p)\,C(\rho_2). \qedhere
\end{align*}
\end{proof}

\begin{proof}[Single-qubit Clifford invariance (iii)]
We note first that, in contrast to the odd-prime case, single-qubit
Cliffords do not in general act as permutations of the Wootters phase
points: $S A_{(q,p)} S^\dagger$ is not a phase-point operator (the sign
of the $Y$ component is inconsistent with every $A_{(q',p')}$), and
similarly for $H$.
Invariance nevertheless holds by passing to Bloch coordinates.
Writing $\rho = \tfrac{1}{2}(I + \vec{r}\cdot\vec{\sigma})$, the Wigner
entries are affine in $\vec{r}$:
\begin{gather*}
  W_\rho(q,p) = \tfrac{1}{4}\bigl(1 + \vec{r}\cdot\vec{t}_{qp}\bigr),\\
  \vec{t}_{qp} = \bigl((-1)^p,\,(-1)^{q+p},\,(-1)^q\bigr),
\end{gather*}
where the four vectors $\vec{t}_{qp}$ form a regular tetrahedron
$T \subset \{\pm 1\}^3$.
Hence for any two single-qubit states,
$\|W_\rho - W_\sigma\|_1
 = \tfrac{1}{4}\sum_{t\in T}\bigl|\langle \vec{r}_\rho - \vec{r}_\sigma,
   t\rangle\bigr| =: N(\vec{r}_\rho - \vec{r}_\sigma)$,
and $C(\rho) = \min_{\vec{r}'\in\mathcal{O}} N(\vec{r}-\vec{r}')$, where
$\mathcal{O}$ is the stabilizer octahedron.
A Clifford $U\in\mathcal{C}_1$ acts on Bloch space as a rotation $R$ in
the octahedral rotation group, which preserves $\mathcal{O}$ and maps the
tetrahedron $T$ onto either $T$ or $-T$.
Since $N$ depends only on the absolute values $|\langle\cdot,t\rangle|$,
it is invariant under $t\mapsto -t$, so $N(R\vec{v}) = N(\vec{v})$ for
every such $R$, and $C(U\rho U^\dagger) = C(\rho)$.
For $n\geq 2$ the relevant norm involves products of tetrahedral frames
and $\Wfree^{(n)}$ contains entangled vertices, and the argument fails;
indeed the local gate $S\otimes I$ changes $C$
(\Cref{rem:cliff_equiv}).
\end{proof}

\begin{proof}[Lipschitz (iv)]
Let $W_f$ be optimal for $\sigma$.
$C(\rho) \leq \|W_\rho-W_f\|_1 \leq \|W_\rho-W_\sigma\|_1 + C(\sigma)$.
\end{proof}

\begin{proof}[Restricted monotonicity (v)]
Since the Wigner transform is a linear bijection, every CPTP map $\Phi$
induces a linear map $T_\Phi$ on phase space with
$W_{\Phi(\rho)} = T_\Phi W_\rho$, defined by
$[T_\Phi]_{\alpha\beta} = W_{\Phi(A_\beta)}(\alpha)$.
Trace preservation gives $\sum_\alpha [T_\Phi]_{\alpha\beta} = 1$ for
every $\beta$, so if the representation is nonnegative
($[T_\Phi]_{\alpha\beta}\geq 0$, the hypothesis), $T_\Phi$ is
column-stochastic and therefore $\ell_1$-contractive:
$\|T_\Phi v\|_1 \leq \|v\|_1$ for every $v$.
Pauli channels satisfy the hypothesis because Pauli conjugations act as
phase-space translations (permutations of the $A_\alpha$), and
computational-basis dephasing is a convex mixture of such conjugations.
Discarding a qubit is the partial trace, which on Wigner functions is
marginalization, $W_{\tr_2\rho}(\alpha)=\sum_\beta W_\rho(\alpha,\beta)$
(because $\sum_\beta A_\beta=2I$ for a single qubit), hence also
stochastic; measure-and-discard is dephasing followed by discarding.
Since $\Phi$ is free it maps $\mathrm{conv}(\Stab_n)$ into itself, so
$T_\Phi(\Wfree)\subseteq\Wfree$.
For $W_f$ optimal for $\rho$,
$C(\Phi(\rho)) \leq \|T_\Phi(W_\rho - W_f)\|_1 \leq \|W_\rho - W_f\|_1
= C(\rho)$.
We emphasize that general Clifford unitaries are \emph{not} in the
hypothesis class: their induced maps $T_\Phi$ have negative entries, are
not $\ell_1$-contractive, and indeed increase $C$
(\Cref{rem:not_monotone}).
\end{proof}

\section{Primal--Dual Formulation}
\label{app:duality}

Let $\{W_i\}_{i=1}^N$ be the Wigner functions of all pure stabilizer states.
The primal LP is:
\begin{alignat}{2}
  &\text{min}  &\quad& \textstyle\sum_k t_k \notag \\
  &\text{s.t.} &&
    W_\rho(k) - \textstyle\sum_i \lambda_i W_i(k) \leq t_k, \notag \\
  &&& -W_\rho(k) + \textstyle\sum_i \lambda_i W_i(k) \leq t_k, \notag \\
  &&& \textstyle\sum_i \lambda_i = 1,\quad \lambda_i,t_k \geq 0. \notag
\end{alignat}

\begin{proof}[Proof of \Cref{thm:dual}]
Introducing dual variables $u_k,v_k\geq 0$ for the inequalities and
$\mu\in\mathbb{R}$ for the equality, the Lagrangian gives feasibility
conditions $u_k+v_k=1$ and $\mu \geq \langle S, W_i\rangle$ for all $i$,
where $S_k := u_k - v_k$.
Since $u_k+v_k=1$ and $u_k,v_k\geq 0$, we have $|S_k|\leq 1$.
Setting $\mu = \max_{W_f\in\Wfree}\langle S,W_f\rangle$ and using
strong duality (primal is feasible and bounded) gives~\eqref{eq:dual}.
\end{proof}

\section{Phase-Space Structural Constraints}
\label{app:structure}

\begin{proposition}[Basic Wigner properties]
For any two-qubit state $\rho$:
(1) $\sum_\alpha W_\rho(\alpha) = 1$.
(2) $-\tfrac18 \leq W_\rho(\alpha) \leq \tfrac{2+\sqrt3}{8}$, since the
eigenvalues of each $A_\alpha$ are $-\tfrac12,-\tfrac12,\tfrac{2-\sqrt3}{2},
\tfrac{2+\sqrt3}{2}$.
(3) For pure states: $\sum_\alpha W_\rho(\alpha)^2 = 1/4$.
\end{proposition}

Random sampling finds pure two-qubit states with as many as $9$ negative
entries; we do not know the maximum.

\end{document}